\documentclass[11pt,a4paper,reqno]{amsart}
\usepackage{amsthm,amsmath,amsfonts,amssymb,amsxtra,appendix,bookmark,euscript,delarray,dsfont,bm,mathrsfs,mathtools}

\usepackage[latin1]{inputenc}
\hypersetup{hidelinks}

\theoremstyle{plain}
\newtheorem{theorem}{Theorem}[section]
\newtheorem{lemma}[theorem]{Lemma}

\newtheorem{proposition}[theorem]{Proposition}
\theoremstyle{definition}

\theoremstyle{remark}
\newtheorem{remark}[theorem]{Remark}

\numberwithin{equation}{section}

\DeclareMathAlphabet{\mathpzc}{OT1}{pzc}{m}{it}
\def\eps{\varepsilon}

\def\N {\mathbb{N}}

\def\R {\mathbb{R}}
\def\S {\mathbb{S}}
\def\Z {\mathbb{Z}}
\newcommand\1{{\ensuremath {\mathds 1} }}

\newcommand{\bp}{\mathbf{p}}
\newcommand{\bx}{\mathbf{x}}
\newcommand{\by}{\mathbf{y}}
\newcommand{\bz}{\mathbf{z}}
\newcommand{\cH}{\mathcal{H}}
\newcommand{\sx}{\textup{x}}

\newcommand{\cE}{\mathcal{E}}

\newcommand{\bDelta}{{\mbox{$\triangle$}\hspace{-8.0pt}\scalebox{0.8}{$\triangle$}}}

\newcommand{\scriptbDelta}{\scalebox{0.8}{$\bDelta$}}

\title{Lieb--Thirring inequalities for wave functions vanishing on the diagonal set}

\author[S. Larson]{Simon LARSON}
\address{Department of Mathematics, California Institute of  Technology, Pasadena, CA 91125, USA}
\email{larson@caltech.edu}

\author[D. Lundholm]{Douglas LUNDHOLM}
\address{Department of Mathematics, Uppsala University,
Box 480, SE-751 06 Uppsala, Sweden}
\email{douglas.lundholm@math.uu.se}

\author[P. T. Nam]{Phan Th\`anh NAM}
\address{Department of Mathematics, LMU Munich, Theresienstrasse 39, 80333 Munich, Germany}
\email{nam@math.lmu.de}


\begin{document}

\begin{abstract} We propose a general strategy to derive Lieb--Thirring inequalities for scale-covariant quantum many-body systems. As an application, we obtain a generalization of the Lieb--Thirring inequality to wave functions vanishing on the diagonal set of the configuration space, without any statistical assumption on the particles. 
\end{abstract}

\maketitle
\setcounter{tocdepth}{2}
\tableofcontents

\section{Introduction}

The celebrated Lieb--Thirring inequality states that the expected kinetic energy of a free Fermi gas is bounded from below by its semiclassical approximation up to a universal factor, namely
\begin{equation} \label{eq:LT}
	\biggl\langle \Psi_N,  \sum_{i=1}^N (-\Delta_{\bx_i})^s \Psi_N\biggr\rangle 
	\ge K \int_{\R^d} \varrho_{\Psi_N}(\bx)^{1+2s/d} d\bx.  
\end{equation}
Here $\Psi_N$ is an $N$-particle wave function in $L^2((\R^d)^N)$, 
normalized so that $\|\Psi_N\|_{L^2(\R^{dN})}=1$ and thus encoding in its squared amplitude 
a probability distribution
for particle positions $\sx = (\bx_1,\ldots,\bx_N)$, $\bx_j \in \R^d$,
with one-body density 
\begin{equation*} 
	\varrho_{\Psi_N}(\bx) := \sum_{j=1}^N \int_{\R^{d(N-1)}} |\Psi(\bx_1,\dots,\bx_{j-1},\bx,\bx_{j+1},\dots,\bx_N)|^2 \prod\limits_{i\ne j} d\bx_i,
\end{equation*} 
and, crucially, subject to the anti-symmetry 
\begin{equation} \label{eq:Pauli}
	\Psi_N(\bx_1,\ldots,\bx_i,\ldots,\bx_j,\ldots,\bx_N) = - \Psi_N(\bx_1,\ldots,\bx_j,\ldots,\bx_i,\ldots,\bx_N), \quad \forall i\ne j.
\end{equation}
This is Pauli's exclusion principle for fermions\footnote{
Here we ignore the spin of particles for simplicity (in our analysis the effect of the spin is mathematically trivial).}.
Replacing the minus sign in \eqref{eq:Pauli} by a plus sign defines bosonic particles,
while if the particles are non-identical, i.e.\ distinguishable, no exchange symmetry may be imposed.

The inequality \eqref{eq:LT} was first proved by Lieb and Thirring in 1975 for the case $s=1$ relevant to non-relativistic particles \cite{LieThi-75,LieThi-76}, and extended by Daubechies in 1983 to general $s>0$, 
thus also including the relativistic case $s=1/2$ \cite{Daubechies-83}. 
The constant $K=K(d,s)>0$ is independent of $N$ and $\Psi_N$ (see \cite{FraHunJexNam-18} for the best known value of $K$). 

The Lieb--Thirring inequality is a beautiful combination of the uncertainty and exclusion principles
of quantum mechanics, and has also been very actively studied in the mathematical literature
from the dual perspective of estimation of eigenvalues of one-body
Schr\"odinger operators (see e.g.\ \cite{LieSei-09,Laptev-12} for reviews). Historically, the Lieb--Thirring inequality was invented to give a short, elegant proof of the stability of ordinary non-relativistic matter with Coulomb forces \cite{LieThi-75}. 
In that context it is well known that stability of the \emph{first} kind, i.e.\ that the ground state energy of the Coulomb system is finite, follows easily from some sort of the uncertainty principle (e.g. Sobolev's inequality). On the other hand, the stability of the \emph{second} kind, that the ground state energy does not diverge faster than the number of particles, is much more subtle: for this the fermionic nature of particles is crucial. In fact, the stability of the second kind fails for bosonic 
(or distinguishable) charged systems \cite{Dyson-67}.  

Without the anti-symmetry condition \eqref{eq:Pauli}, the Lieb--Thirring inequality \eqref{eq:LT} fails and the best one can get is 
the Gagliardo-Nirenberg-Sobolev inequality 
\begin{equation} \label{eq:LT-boson}
	\biggl\langle \Psi_N, \sum_{i=1}^N (-\Delta_{\bx_i})^s \Psi_N\biggr\rangle \ge K N^{-2s/d} \int_{\R^d} \varrho_{\Psi_N}(\bx)^{1+2s/d} d\bx  
\end{equation}
(see e.g.~\cite{LunNamPor-16}). The emergence of the factor $N^{-2s/d}$ can be seen by considering the bosonic trial state $\Psi_N=u^{\otimes N}$ (whose density is $\varrho_{\Psi_N}(x)=N|u(x)|^2$).  This factor is small when $N$ becomes large, making \eqref{eq:LT-boson} not very useful in applications. 

Note that Pauli's exclusion principle \eqref{eq:Pauli} implies that the wave function $\Psi_N$ vanishes on the diagonal set
\begin{equation} \label{eq:diagonal}
	\bDelta := \bigl\{(\bx_1,\ldots,\bx_N) \in (\R^d)^{N}: \bx_i=\bx_j \text{ for some } i\neq j \bigr\},
\end{equation}
namely there is zero probability for two quantum particles to occupy a common single position in the configuration space. 
	
In this paper, we want to address the following 
	
\smallskip

\noindent
{\bf Question:} Does the Lieb--Thirring inequality \eqref{eq:LT} remain valid if the 
anti-symmetry assumption \eqref{eq:Pauli} is replaced by the weaker condition 
$\Psi_N\rvert_{\scriptbDelta}=0\,?$

\smallskip

We will show that the answer is {\bf yes} if and only if $2s>d$. 
In fact, $2s>d$ is the optimal condition for the vanishing assumption 
$\Psi_N\rvert_{\scriptbDelta}=0$ to be non-trivial (heuristically this follows from Sobolev's embedding 
$  H^{s}(\R^d)\subseteq C(\R^d)$ for $2s>d$). 
The precise statement of our result and its consequences will be presented in the next section.

\section{Main results}

Recall that for every $s>0$ (not necessarily an integer) the operator $(-\Delta)^s$ 
on $L^2(\R^d)$ is defined as the multiplication operator $|\bp|^{2s}$ 
in Fourier space, namely
$$
	\bigl[(-\Delta)^s f\bigr]^{\!\wedge\!}(\bp)= |\bp|^{2s} \widehat{f}(\bp), 
	\qquad \widehat{f}(\bp):= \frac{1}{(2\pi)^{d/2}}\int_{\R^d}f(\bx)e^{-i\bp\cdot\bx}\,d\bx.
$$
The associated space $H^s(\R^d)$ is a Hilbert space with norm
$$
	\|u\|_{H^s(\R^d)}^2 := \|u\|_{L^2(\R^d)}^2 + \|u\|_{\dot H^s(\R^d)}^2,
	\qquad
	\|u\|_{\dot H^s(\R^d)}^2 := \langle u, (-\Delta)^s u\rangle.
$$
The $N$-particle space $H^s(\R^{dN})$ is defined in the same way. 
Let us denote the subspace of functions vanishing on the diagonal set $\bDelta$ in \eqref{eq:diagonal} by 
$$ 
	\cH^{s,N}(\R^d):= \overline{\bigl\{ \Psi_N \in C_c^\infty(\R^{dN}): \Psi_N\rvert_{\scriptbDelta}=0 \bigr\} }^{\,H^s(\R^{dN})}.
$$

Our main result is 
\begin{theorem}[Lieb--Thirring inequality for wave functions vanishing on diagonals] \label{thm:main}
	Let $2s>d\ge 1$. Then for every $N \ge 1$ 
	and $\Psi_N\in\cH^{s,N}(\R^d)$, with $\|\Psi_N\|_{L^2(\R^{dN})}=1$, we have
	\begin{equation} \label{eq:LT-diagonal}
		\biggl\langle \Psi_N,  \sum_{i=1}^N (-\Delta_{\bx_i})^s\Psi_N \biggr\rangle
		\ge C\int_{\R^d} \varrho_{\Psi_N}(\bx)^{1+2s/d} \,d\bx.
	\end{equation}
	Here $C=C(d,s)>0$ is a universal constant independent of $N$ and $\Psi_N$. 
\end{theorem}

We have some immediate remarks. 

1. The condition $2s>d$ in Theorem \ref{thm:main} is optimal. If $2s\le d$, then
$$\cH^{s,N}(\R^d)=H^s(\R^{dN})$$
by the relatively small size, i.e. the large codimensionality, of the diagonal set (see Appendix~\ref{app:subspaces}) 
and thus the Lieb--Thirring inequality fails. 

2. For $d=1$ and $s=1$, it is well known that a \emph{symmetric} wave function which vanishes on the diagonal set is equal to an anti-symmetric wave function up to multiplication 
by an appropriate sign function~\cite{Girardeau-60}, and hence \eqref{eq:LT-diagonal} 
reduces to the usual Lieb--Thirring inequality~\cite{LieThi-76} in this case. 
However, when $d>1$ this boson-fermion correspondence is no longer available and our result is new. 
Furthermore, one may consider \emph{hard-core} bosons
defined by the higher-order vanishing around diagonals
\begin{equation} \label{eq:cH0}
	\cH_0^{s,N}(\R^d):= \overline{\bigl\{ \Psi_N \in C_c^\infty(\R^{dN} \setminus \bDelta) \bigr\} }^{\, H^s(\R^{dN})},
\end{equation}
and subject to symmetry. For large enough order $2s>d$ there is even for $d=1$ a non-trivial difference 
between these spaces, and our result assumes only the weaker vanishing conditions imposed by 
$\cH^{s,N}(\R^d)$ (see Appendix~\ref{app:subspaces} for some further remarks).

3. Theorem \ref{thm:main} verifies a conjecture in \cite[page 1362]{LunNamPor-16} that the Lieb--Thirring inequality~\eqref{eq:LT-diagonal} holds for all wave functions in the form domain of the interaction potential 
$$W_s(\sx) := \sum_{1\le i<j \le N} |\bx_i-\bx_j|^{-2s}, \qquad \sx=(\bx_1,\ldots,\bx_N)\in (\R^d)^N.$$
In fact, we have
(again, see Appendix~\ref{app:subspaces} for details)
\begin{equation} \label{eq:interaction-form-domain}
	\biggl\{  \Psi_N \in H^s(\R^{dN}):  \int_{\R^{dN}} W_s(\sx) |\Psi_N(\sx)|^2 \,d\sx  < \infty \biggr\} 
	\subseteq \cH_0^{s,N}(\R^d)
	\subseteq \cH^{s,N}(\R^d),
\end{equation}
by the singular nature of the potential at the diagonals. We may think of the potential $W_s$ as defining (by Friedrichs extension)
a one-parameter family of non-negative and scale-covariant 
(scaling homogeneously to degree $-2s$)
interacting $N$-body Hamiltonian operators
$$
	H_\beta := \sum_{j=1}^N (-\Delta_{\bx_j})^s + \beta W_s,
	\qquad \beta \ge 0.
$$
(In the case $\beta=0$, then $H_\beta$ is still defined on the quadratic form domain of $W_s$.) One may then ask about the best constant $C(\beta)\ge 0$ in the bound
\begin{equation} \label{eq:Hbeta-LT}
	\langle \Psi_N, H_\beta \Psi_N \rangle \ge C(\beta)\int_{\R^d}\varrho_{\Psi_N}(\bx)^{1+2s/d}\,d\bx,
\end{equation}
i.e.\ a Lieb--Thirring inequality (generalized uncertainty principle) for $H_\beta$. The case $\beta>0$ was treated in \cite{LunPorSol-15,LunNamPor-16}, while
our setting here concerns the limit $\beta \to 0$ of zero-range/contact interaction.
A crucial difference is the strength of the interaction term, 
which is of order $\beta N^2$ and thus provides a large repulsive energy
for fixed $\beta > 0$, while for $\beta \ll 1/N$ it ought to be much weaker than the kinetic term. Nevertheless, for $2s > d$ the potential $W_s$ is singular enough to impose the vanishing condition 
at $\bDelta$, and Theorem~\ref{thm:main} thus implies that the limiting constant is positive $C(0)>0$. Furthermore, for $2s=d$ we have $C(0)=0$ which follows from a straightforward argument through contradiction using the density of $C_c^\infty(\R^{dN}\setminus\bDelta)$ in $H^s(\R^{dN})$; again see our first remark and Appendix~\ref{app:subspaces}. Together with the case $2s<d$ that was treated in~\cite{LunNamPor-16} and for which also $C(0)=0$, this settles the question concerning which $H_\beta$ satisfy the Lieb--Thirring inequality \eqref{eq:Hbeta-LT} with $C(\beta)>0$.

4. The original proof of the Lieb--Thirring inequality \cite{LieThi-75,LieThi-76} is based on the following operator bound  
\begin{equation} \label{eq:Pauli-ope}
0\le \gamma_{\Psi_N}^{(1)}\le \1
\end{equation}
which is a consequence of Pauli's exclusion principle \eqref{eq:Pauli}. 
Here $\gamma_{\Psi_N}^{(1)}$ is the one-body density matrix of $\Psi_N$, 
a trace-class operator on $L^2(\R^d)$ with kernel
$$
	\gamma_{\Psi_N}^{(1)}(\bx;\bx') 
	= \sum_{j=1}^N \int_{\R^{d(N-1)}} \Psi_N(\bx_1,\ldots,\bx_j=\bx,\ldots,\bx_N) \overline{\Psi_N(\bx_1,\ldots,\bx_j=\bx',\ldots,\bx_N)} \prod_{k\neq j} d\bx_k.
$$
However, unlike the full anti-symmetry condition \eqref{eq:Pauli}, 
the vanishing condition $\Psi_N\rvert_{\scriptbDelta}=0$ alone is not known to be 
sufficient to ensure the operator inequality \eqref{eq:Pauli-ope}, 
and therefore the original proof in \cite{LieThi-75,LieThi-76} 
as well as subsequent proofs based on \eqref{eq:Pauli-ope} (e.g.\ Rumin's method~\cite{Rumin-11}) do not apply. 

\medskip

Our result is in fact more general than as previously formulated. More precisely, define 
for any $k \ge 2$ the diagonal set of $k$-particle coincidences
\begin{equation} \label{eq:k-diagonal}
	\bDelta_k := \bigl\{(\bx_1,\ldots,\bx_N) \in (\R^d)^{N}: \bx_{j_1}= \ldots = \bx_{j_k} \text{ for distinct indices } j_1,\ldots,j_k \bigr\},
\end{equation}
and the corresponding space of $N$-particle wave functions with a vanishing condition on $\bDelta_k$
$$
	\cH^{s,N}_k(\R^d) := \overline{\bigl\{ \Psi_N \in C_c^\infty(\R^{dN}): \Psi_N\rvert_{\scriptbDelta_k}=0 \bigr\} }^{\, H^s(\R^{dN})}.
$$
We have 
\begin{theorem}[Lieb--Thirring inequality for wave functions vanishing on $k$-diagonals] \label{thm:main-general} 
	Let $d\ge 1$, $k \ge 2$ and $2s>d(k-1)$. Then for every $N \ge 1$ and every 
	$\Psi_N \in\cH^{s,N}_k(\R^d)$, with $\|\Psi_N\|_{L^2(\R^{dN})}=1$, we have
	\begin{equation} \label{eq:LT-diagonal-general}
		\biggl\langle \Psi_N,  \sum_{i=1}^N (-\Delta_{\bx_i})^s\Psi_N \biggr\rangle
		\ge C\int_{\R^d} \varrho_{\Psi_N}(\bx)^{1+2s/d} \,d\bx.
	\end{equation}
	Here $C=C(d,s,k)>0$ is a universal constant independent of $N$ and $\Psi_N$. 
\end{theorem}

\smallskip

The proof of Theorem~\ref{thm:main-general} occupies the rest of the paper. Our proof is based on a general strategy of deriving Lieb--Thirring inequalities for wave functions satisfying some partial exclusion properties, which was proposed by Lundholm and Solovej in \cite{LunSol-13a} and developed further in \cite{FraSei-12,LunSol-13b,LunSol-14,LunPorSol-15,LunNamPor-16,LarLun-18,Nam-18,LunSei-18,Lundholm-18}. 
We will quickly review this strategy in Section~\ref{sec:general-strategy} for the reader's convenience, following the simplification by Lundholm, Nam and Portmann~\cite{LunNamPor-16}. 

The main new ingredient is a local version of the exclusion principle using the vanishing condition on the diagonal set. In Section \ref{sec:reduction}, we will discuss a very useful reduction of the desired local exclusion to 
simply the positivity of a local energy using 
the scale-covariance of the kinetic operator $(-\Delta)^s$. This step refines and generalizes a recent bootstrap argument for the energy of ideal anyons by Lundholm and Seiringer \cite{LunSei-18}. In Section \ref{sec:Poincare}, the remaining crucial fact that the local energy eventually becomes positive with increasing particle number will be settled by means of a new many-particle Poincar\'e inequality.  Some standard and non-standard results on relevant function spaces are collected in the appendices for completeness. 

We stress that our method will also work for any other deformations of the Laplacian which retain similar
positivity and scale-covariance properties, including other types of point interactions 
as well as particles subject to intermediate statistics (ideal anyons) in one and two dimensions.

\medskip

\noindent\textbf{Acknowledgments.} S.L. and D.L. thank John Andersson for helpful discussions.  S.L. acknowledges financial support from the Swedish Research Council grant no.~2012-3864.
D.L. acknowledges financial support by the grant no.~1804 from the G\"oran Gustafsson Foundation
and the Swedish Research Council grant no. 2013-4734.  Part of this work was carried out during the Conference ``Eigenvalues and Inequalities" at the Institut Mittag-Leffler, Stockholm, May 2018.

\section{General strategy of deriving Lieb--Thirring inequalities} \label{sec:general-strategy}

In the following we will summarize the proof of the usual Lieb--Thirring inequality~\eqref{eq:LT} 
for fermionic wave functions, mainly following the simplified representation in  \cite{LunNamPor-16}. The starting point is the following obvious localization formula: if $\{\Omega\}$ is a collection of disjoint subsets of $\R^d$, then 
\begin{equation} \label{eq:operator-localization}
(-\Delta)^s_{|\R^d} \ge \sum_{\Omega} (-\Delta)^s_{|\Omega},
\end{equation}
where the Neumann localization $(-\Delta)^s_{|\Omega}$ is defined via the quadratic form 
(Sobolev seminorm)
$$\langle u, (-\Delta)^s_{|\Omega} u\rangle= \|u\|_{\dot{H}^s(\Omega)}^2 	:= \!\begin{cases}
	\sum\limits_{|\alpha|=m} \dfrac{m!}{\alpha!} \int_{\Omega} |D^\alpha u|^2 & \text{if~}s=m, \\[0.2cm]  
	c_{d,\sigma}\sum\limits_{|\alpha|=m} \dfrac{m!}{\alpha!} \iint_{\Omega\times \Omega}\! \dfrac{|D^\alpha u (\bx) 
	-D^\alpha u (\by)|^2}{|\bx-\by|^{d+2\sigma}}d\bx d\by & \text{if~} s=m+\sigma, \\
	\end{cases}
$$
for all $u\in H^s(\R^d)$, with $m \in \N_0$, $\alpha \in \N_0^d$ multi-indices, 
$D^\alpha$ corresponding derivatives,
and 
$$
 0<\sigma<1, \quad c_{d,\sigma} :=\frac{2^{2\sigma-1}}{\pi^{d/2}} \frac{\Gamma((d+2\sigma)/2)}{|\Gamma(-\sigma)|}.
$$
Consequently, for any $N$-body wave function $\Psi_N \in H^s(\R^{dN})$ we have
\begin{equation} \label{eq:localization}
	\cE_{\R^d}[\Psi_N] \ge \sum_\Omega {\mathcal E}_\Omega[\Psi_N],
\end{equation}
where the expected local energy on $\Omega$ is
\begin{equation} \label{eq:def-EQ}
	\cE_\Omega[\Psi_N] := \biggl\langle \Psi_N, \sum_{j=1}^N (-\Delta_{\bx_j})^s_{|\Omega} \Psi_N \biggr\rangle = \sum_{j=1}^N \int_{\R^{d(N-1)} } \| \Psi_N\|^2_{\dot H^s_{\bx_j}(\Omega)} \prod_{\ell \ne j} d \bx_{\ell} 
\end{equation}

Next, we have the following three key tools \cite[Lemmas 8, 11, 12]{LunNamPor-16}.

\begin{lemma}[Local uncertainty]\label{lem:local-uncertainty}
	Let $d\ge 1$ and $s>0$. Let $\Psi_N$ be a wave function in $H^s(\R^{dN})$ 
	for arbitrary $N \ge 1$
	and let $Q$ be an arbitrary cube in $\R^d$. Then 
	\begin{equation} \label{eq:local-uncertainty}
		\cE_Q[\Psi_N] \ge \frac{1}{C}\frac{\int_Q \varrho_{\Psi_N}^{1+2s/d}}{\Bigl(\int_Q \varrho_{\Psi_N} \Bigr)^{2s/d}} 
		- \frac{C}{|Q|^{2s/d}} \int_Q \varrho_{\Psi_N}. 
\end{equation}
\end{lemma}
Hereafter, 	$C=C(d,s)>0$ denotes a universal constant (independent of $N$, $\Psi_N$ and $Q$).

Lemma \ref{lem:local-uncertainty} can be interpreted as a local version of the  lower bound \eqref{eq:LT-boson} (the negative term appears due to the lack of Dirichlet boundary condition).

\begin{lemma}[Local exclusion for fermions] \label{lem:local-exclusion-fermions} 
	Let $d\ge 1$ and $s>0$. Let $\Psi_N$ be a fermionic wave function in $H^s(\R^{dN})$ satisfying \eqref{eq:Pauli} for $N \ge 2$ and let $Q$ be an arbitrary cube in $\R^d$. Then 
\begin{equation} \label{eq:local-exclusion}
		\cE_Q[\Psi_N] \ge  C |Q|^{-2s/d}  \biggl[ \int_Q \varrho_{\Psi_N}(\bx)\,d\bx - q \biggr]_+,
\end{equation}
	where
	$q := \# \{ \textup{multi-indices}\ \alpha \in \N_0^d : 0 \le |\alpha|<s \}$.
\end{lemma}

In the non-relativistic case $s=1$,
Lemma~\ref{lem:local-exclusion-fermions} simply states that as soon as there is
more than one particle on $Q$ the energy must be strictly positive,
and furthermore that it grows at least linearly with the number of particles.
Such a weak formulation of the exclusion principle was used by Dyson and Lenard in their first proof of the stability of matter~\cite{DysLen-67}, while its general applicability in the above format was noted by Lundholm and Solovej in~\cite{LunSol-13a,LunSol-13b}.

\begin{lemma}[Covering lemma] \label{lem:covering} Let $0\le f\in L^1(\R^d)$ be a function with compact support such that $\int_{\R^d} f \ge \Lambda>0$. Then the support of $f$ can be covered by a collection of disjoint cubes $\{Q\}$ in $\R^d$ such that 
\begin{equation} \label{eq:covering-0} \int_{Q} f \le \Lambda, \quad \forall Q
 \end{equation}
and
\begin{equation} \label{eq:covering}
			\sum_{Q} \frac{1}{|Q|^{\alpha}} \Biggl( \biggl[\int_{Q} f - q \biggr]_+ 
				- b \int_{Q} f \Biggr) \ge 0
		\end{equation}
		for all $\alpha>0$ and $0\le q< \Lambda 2^{-d}$, where
		$$
			b:= \biggl( 1- \frac{2^dq}{\Lambda} \biggr) \frac{2^{d\alpha}-1}{2^{d\alpha} + 2^d - 2}>0. 
		$$
\end{lemma}

\noindent 
{\bf  Conclusion of \eqref{eq:LT}.}  Let $q$ be as in Lemma \ref{lem:local-exclusion-fermions} and let $\Lambda=2^d q+1$. 
If $N\le \Lambda$, then~\eqref{eq:LT} follows immediately from \eqref{eq:LT-boson}, whose proof is similar to (indeed simpler than) that of Lemma \ref{lem:local-uncertainty}. 
If $N> \Lambda$, then we can apply Lemma \ref{lem:covering} with $f=\varrho_{\Psi_N}$
(by standard approximation we may reduce to compact support), 
$\alpha=2s/d$, and obtain a collection of disjoint cubes $\{Q\}$. Combining with \eqref{eq:localization}, \eqref{eq:local-uncertainty} and \eqref{eq:local-exclusion} we obtain
\begin{align*}
(\eps+1) \cE_{\R^d}[\Psi_N] &\ge \eps \sum_Q \left[ \frac{1}{C_1}\frac{\int_Q \varrho_{\Psi_N}^{1+2s/d}}{\bigl(\int_Q \varrho_{\Psi_N} \bigr)^{2s/d}} - \frac{C_1}{|Q|^{2s/d}} \int_Q \varrho_{\Psi_N}\right] \\
& \quad + \sum_Q C_2 |Q|^{-2s/d}  \biggl[ \int_Q \varrho_\Psi(x)\,dx - q \biggr]_+  \\
&\ge \frac{\eps}{C_1} \frac{\int_{\R^d} \varrho_{\Psi_N}^{1+2s/d}}{\Lambda^{2s/d}} 
\end{align*}
for any fixed constant $\eps>0$ satisfying $\eps C_1\le C_2 b$.  Thus  \eqref{eq:LT} holds true. 

\smallskip

As we can see from the above strategy, the only place where the anti-symmetry \eqref{eq:Pauli} plays a role is the local exclusion bound in Lemma \ref{lem:local-exclusion-fermions}. Extending this result to the weaker condition $\Psi_N\rvert_{\scriptbDelta}=0$ is the main task of our proof below.

\section{Reduction of local exclusion} \label{sec:reduction}

In this section, we prove a very useful observation, that allows to reduce the local exclusion \eqref{eq:local-exclusion} to the positivity of the local energy, using the scale-covariance of the kinetic energy. 
This step is inspired by the recent work of Lundholm and Seiringer \cite{LunSei-18} on the energy of ideal anyons. We formulate it abstractly as follows:

\begin{lemma}[Covariant energy bound]\label{lem:energy} 
	Assume that to any $n\in \mathbb{N}_0$ and any cube $Q\subset \R^d$ there is associated a non-negative number (`energy') $E_n(Q)$ satisfying the following properties, for some constant $s>0$: 
	\begin{itemize}
		\item (scale-covariance)
			$E_n(\lambda Q) = \lambda^{-2s} E_n(Q)$ for all $\lambda > 0$;
		\item (translation-invariance)
			$E_n(Q+\bx) = E_n(Q)$ for all $\bx \in \R^d$;
		\item (superadditivity)
			For any collection of disjoint cubes $\{Q_j\}_{j=1}^J$ such that their union is a cube, 
			$$E_n\Bigl(\bigcup_{j=1}^J Q_j\Bigr) \ge \min_{\{n_j\} \in \mathbb{N}_0^J \, s.t.\, \sum_j n_j = n}\ \sum_{j=1}^J E_{n_j}(Q_j) ;$$
		\item (a priori positivity) There exists $q \ge 0$ such that $E_n(Q) > 0$ for all $n \ge q$.
	\end{itemize}
	Then there exists a constant $C>0$ independent of $n$ and $Q$ such that
	\begin{equation}\label{eq:energy-bound}
		E_n(Q) \ge C |Q|^{-2s/d} n^{1+2s/d}, \quad \forall n\ge q.
	\end{equation}
\end{lemma}

\begin{proof}
	Note that for $q \le n \le N$,
	\eqref{eq:energy-bound} holds for some $C=C_N > 0$ by the a priori positivity. The main point here is to remove the $N$-dependence of the constant. 

	Denote $E_n := E_n(Q_0)$ with $Q_0=[0,1]^d$. Assume by induction in $N$ that
	\begin{equation}\label{eq:energy-bound-induc}
		E_n\ge C n^{1+2s/d}, \quad \forall q \le n \le N-1
	\end{equation}
	with a uniform constant $C>0$ and consider $n=N$.
	Split $Q_0$ into $2^d$ subcubes of half side length and obtain by  the superadditivity, translation-invariance and scale-covariance 
	\begin{equation}\label{eq:energy-split}
		E_{N} \ge 2^{2s} \min_{\{n_j\} \, s.t.\, \sum_j n_j = N} \ \sum_{j=1}^{2^d} E_{n_j}.
	\end{equation}
	
	Consider a configuration $\{n_j\}\subset \mathbb{N}_0^{2^d}$ such that the minimum in  \eqref{eq:energy-split} is attained. 
	The a priori positivity $E_N>0$ ensures that none of the $n_j$ can be $N$
	(in the same way we deduce that $E_0 = 0$). 
	Assume that there exist exactly $M$ numbers $n_j < q$ with $0 \le M \le 2^d$. Then 
	$$
		\sum_{n_j \ge q} 1 = 2^d - M 
		\quad \text{and} \quad
		\sum_{n_j\ge q} n_j = N - \sum_{n_j<q} n_j \ge N - qM.
	$$
	Therefore, from \eqref{eq:energy-split}, \eqref{eq:energy-bound-induc} and H\"older's inequality we deduce that  
	\begin{equation}\label{eq:energy-split-2}
		E_{N} \ge C 2^{2s} \sum_{n_j \ge q} n_j^{1+2s/d} 
		\ge C 2^{2s}  \frac{\Bigl( \sum_{n_j\ge q} n_j \Bigr)^{1+2s/d}}{\Bigl(\sum_{n_j\ge q} 1\Bigr)^{2s/d}}
		\geq  C N^{1+2s/d}  \frac{ (1 -  qM N^{-1})^{1+2s/d} }{ (1 -M 2^{-d})^{2s/d} } 
	\end{equation}
	with the same constant $C$ as in \eqref{eq:energy-bound-induc}. If we take 
	$$
		N \ge q2^{d} \Bigl(1+\frac{d}{2s}\Bigr),
	$$
	so that also $qMN^{-1} \le 1$,
	then by Bernoulli's inequality
	$$
		(1-qMN^{-1})^{1+d/(2s)} \ge 1- qMN^{-1} \Bigl(1+ \frac{d}{2s}\Bigr) \ge 1-M2^{-d},
	$$
	and hence \eqref{eq:energy-split-2} reduces to 
	\begin{equation}\label{eq:energy-split-3}
		E_{N} \ge CN^{1+2s/d}
	\end{equation}
	with the same constant $C$ as in~\eqref{eq:energy-bound-induc}. 

	By induction we obtain~\eqref{eq:energy-split-3} for all $N\ge q$, with a constant $C$ independent of $N$. This is the desired bound~\eqref{eq:energy-bound} for the unit cube $Q_0$. 
	The result for the general cube follows from scale-covariance and translation-invariance. 
\end{proof}

\begin{remark}
	It is in fact also possible to allow for $E_n < 0$ for finitely many $n>0$ 
	in Lemma~\ref{lem:energy},
	under a small refinement of the assumption of a priori positivity.
	It is sufficient that there exists $q > 0$ and $c > 1$ 
	such that for all $n \ge q$
	\begin{equation}\label{eq:energy-positivity-refined}
		E_n > c \frac{d}{2s} 2^{d+2s} E_-,
		\qquad E_- := \max_{0 \le n < q} (-E_n).
	\end{equation}
	Namely, with this assumption, the bound in \eqref{eq:energy-split-2}
	may again be used for all $n_j \ge q$, and one obtains
	$E_N \ge C N^{1+2s/d} f(M/2^d)$,
	$C = \min_{q \le n \le N-1} E_n/n^{1+2s/d}$,
	where the function
	$$
		f(x) := \frac{(1-q2^d N^{-1} x)^{1+2s/d}}{(1-x)^{2s/d}} - \frac{E_- 2^{d+2s}}{C N^{1+2s/d}} x,
		\qquad x \in [0,1),
	$$
	is strictly increasing if $N$ is large enough.
\end{remark}

We will apply the above general bound to the local ground-state energy among wave functions satisfying the vanishing condition on $k$-particle diagonals
\begin{equation} \label{eq:def-EnQ}
	E_N(\Omega) := \inf \Bigl\{ \|\Psi_N\|_{\dot H^{s,N}(\Omega)}^2 : \Psi_N \in \cH^{s,N}_k(\R^d), \| \Psi_N\|_{L^2(\Omega^N)}=1 \Bigr\},
\end{equation}
where we have introduced the `completely localized' kinetic functional 
\begin{align} \label{eq:def-local-kinetic}
	\|\Psi_N\|_{\dot H^{s,N}(\Omega)}^2 := \biggl\langle \1_{\Omega^N} \Psi_N, \sum_{j=1}^N (-\Delta_{\bx_j})^s_{|\Omega} \1_{\Omega^N} \Psi_N \biggr\rangle = \sum_{j=1}^N \int_{\Omega^{N-1}} \|\Psi_N\|_{\dot H^{s}_{\bx_j}(\Omega)}^2 \prod_{\ell \ne j} d \bx_\ell.
\end{align}

Note that $\|\Psi_N\|_{\dot H^{s,N}(\Omega)}^2$ is {\em different} from the functional $\cE_{\Omega}[\Psi_N]$ in \eqref{eq:def-EQ}, and its properties will be crucial to deduce the desired local exclusion for  $\cE_{\Omega}[\Psi_N]$. The seminorm $\|\cdot\|_{\dot H^{s,N}(\Omega)}$ in general contains only some of the terms of  the standard homogeneous Sobolev seminorm $\|\cdot\|_{\dot H^{s}(\Omega^N)}$; 
however, the corresponding norms (i.e.\ the seminorms plus the $L^2$-norm) are actually equivalent modulo $N$-dependent constants, not only globally on $\R^{dN}$ but also locally on $Q^N$ (see Appendix~\ref{app:norms}). 

The superadditivity of the energy $E_N(\Omega)$ follows  from 
the partitioning of the many-body space
and by locality respectively non-negativity of any non-local 
part of the kinetic energy, i.e.~\eqref{eq:operator-localization}.
The method was also used in \cite[Lemma~4.2]{LunSei-18} for anyons.

\begin{lemma}[Superadditivity of $E_n(\Omega)$]\label{lem:superadd}
	Let $\{\Omega_j\}_{j=1}^J$ be a collection of disjoint
	subsets of~$\R^d$ and $\Omega=\cup_{j} \Omega_j$. Then 
	\begin{equation}\label{eq:superadd}
		E_N(\Omega) \geq \min_{\{n_j\} \in \mathbb{N}_0^J \, s.t. \, \sum_{j} n_j=N} \, \sum_{j=1}^J E_{n_j}(\Omega_j).
	\end{equation}
\end{lemma}

\begin{proof} For any partition $A=\{A_j\}_{j=1}^J$ of $\{1, 2, \ldots, N\}$ (i.e.\ the $A_j$ are disjoint subsets of $\{1, 2, \ldots, N\}$ such that $\sum_j |A_j|=N$), we denote by $\1_{A}$ the characteristic function of the set 
$$ 
	\Bigl\{ (\bx_1,\dots,\bx_N) \in (\R^d)^N : \bx_i\in \Omega_j  \Leftrightarrow\, i \in A_j, \text{ for all } i,j\Bigr\}.
$$
Using the operator bound similar to \eqref{eq:operator-localization}
$$
(-\Delta_{\bx_i})^s_{|\Omega} \ge \sum_{j=1}^J (-\Delta_{\bx_i})^s_{|\Omega_j},
$$
the partition of unity
\begin{align}\label{eq:pt-A}
\1_{\Omega^N} = \sum_{A} \1_{A},   
\end{align}
and the fact that $\1_A$ commutes with $ (-\Delta_{\bx_i})^s_{|\Omega_j}$, we can write for any $\Psi\in \cH_{k}^{s,N}(\R^d)$
\begin{align*}
	\|\Psi\|_{\dot{H}^{s,N}(\Omega)}^2 &= \biggl \langle \1_{\Omega^N}\Psi,  \sum_{i=1}^N (-\Delta_{\bx_i})^s_{|\Omega} \1_{\Omega^N} \Psi \biggr\rangle \nonumber \\
	&\ge 
	\sum_{j=1}^J  \biggl \langle  \1_{\Omega^N} \Psi,  \sum_{i=1}^N (-\Delta_{\bx_i})^s_{|\Omega_j}  \1_{\Omega^N} \Psi\biggr\rangle \\
	&= 
	\sum_{j=1}^J \sum_{A}  \biggl \langle \1_A \Psi,  \sum_{i=1}^N (-\Delta_{\bx_i})^s_{|\Omega_j} \1_A \Psi\biggr\rangle\\
	&= 
	\sum_{j=1}^J \sum_{A}  \int_{ \R^{d (N-|A_j|)}} \| \Psi(\,\cdot\,; \sx_{A_j^c})  \|^2_{\dot H^{s,|A_j|}(\Omega_j) } \prod_{\ell\ne j}  \Bigl[ \1_{\Omega_\ell^{|A_\ell|}} (\sx_{A_\ell}) d \sx_{A_\ell}\Bigr]. 
\end{align*}
Here we have introduced the shorthand notation
$$
	(\bx_1,\ldots,\bx_N)=(\sx_{A_j};\sx_{A_j^c}), \quad    \sx_{A_j}= (\bx_\ell)_{\ell \in A_j} \in (\R^{d})^{ |A_j|}.
$$
Since $\Psi\in \cH_{k}^{s,N}(\R^d)$, for a.e. $\sx_{A_j^c}\in \R^{d(N-|A_j|)}$ the function $\Psi(\,\cdot\,;\sx_{A_j^c})$ is in $\cH^{s,|A_j|}_{k}(\R^d)$, and hence 
\begin{align*}
	\| \Psi(\,\cdot\,; \sx_{A_j^c})  \|^2_{\dot H^{s,|A_j|}(\Omega_j) } 	
	& \ge 
	E_{|A_j|}(\Omega_j) \int_{\Omega_j^{|A_j|}} |\Psi(\sx_{A_j}; \sx_{A_j^c})|^2  d \sx_{A_j} \\
	&=
	E_{|A_j|}(\Omega_j) \int_{\R^{d|A_j|}} |\Psi(\sx_{A_j}; \sx_{A_j^c})|^2  \1_{\Omega_j^{|A_j|}} (\sx_{A_j}) d \sx_{A_j}. 
\end{align*}
Thus in summary
\begin{align*}
	\|\Psi\|_{\dot{H}^{s,N}(\Omega)}^2 &\ge \sum_{j=1}^J  \sum_{A} E_{|A_j|}(\Omega_j) \int_{ \R^{dN}} | \Psi |^2 \prod_{\ell=1}^J  \Bigl[ \1_{\Omega_\ell^{|A_\ell|}} (\sx_{A_\ell}) d \sx_{A_\ell} \Bigr]  \\
	&=  
	\sum_{j=1}^J  \sum_{A} E_{|A_j|}(\Omega_j) \langle \Psi, \1_A \Psi\rangle
	= 
	\sum_{A} \Bigl[ \sum_{j=1}^J E_{|A_j|}(\Omega_j) \Bigr] \langle \Psi, \1_A \Psi\rangle\\
	&\ge 
	\Bigl[\hspace{1pt} \min_{\{n_j\} \in \mathbb{N}_0^J \, s.t. \, \sum_{j} n_j=N} \, \sum_{j=1}^J E_{n_j}(\Omega_j) \Bigr] \sum_A \langle \Psi, \1_A \Psi\rangle\\
	&= 
	\Bigl[\hspace{1pt} \min_{\{n_j\} \in \mathbb{N}_0^J \, s.t. \, \sum_{j} n_j=N} \, \sum_{j=1}^J E_{n_j}(\Omega_j) \Bigr] \|\Psi\|^2_{L^2(\Omega^N)}.
\end{align*}
Here in the last identity we have used the partition of unity~\eqref{eq:pt-A} again. This implies the desired estimate~\eqref{eq:superadd}.  
\end{proof}

Now we are ready to prove the reduction of the local exclusion. 

\begin{lemma}[Energy positivity implies local exclusion] \label{lem:sim-local-exclusion}Assume that there exists a constant $q>0$ such that for any cube $Q\subset \R^d$, 
\begin{equation} \label{eq:local-exclusion-positivity}
	E_N(Q) > 0, \quad \forall N\ge q.
\end{equation} 
Then for all $N \ge 1$ and for all wave functions 
$\Psi_N\in \cH^{s,N}_k(\R^d)$, $\|\Psi_N\|_{L^2(\R^{dN})}=1$, we have
\begin{equation} \label{eq:loca-exclusion-vanishing}
	\cE_Q[\Psi_N] \ge  C |Q|^{-2s/d}  \biggl[ \int_Q \varrho_{\Psi_N}(\bx)\,d\bx - q \biggr]_+.
\end{equation} 
Here $C>0$ is a constant independent of $N$, $\Psi_N$ and $Q$. 
\end{lemma}

\begin{proof} Given \eqref{eq:local-exclusion-positivity},
the energy functional $E_{n}(Q)$ defined in \eqref{eq:def-EnQ} verifies all conditions in Lemma~\ref{lem:energy}. 
Therefore, there exists a constant $C>0$ independent of $n$ and $Q$ such that
\begin{equation} \label{eq:EnQ-lower-bound}
	E_n(Q) \ge C |Q|^{-2s/d} n^{1+2s/d} \1_{\{n\ge q\}} \ge C |Q|^{-2s/d} [n-q]_+, \quad \forall n\ge 0. 
\end{equation}

Now we adapt the localization method in the proof of Lemma \ref{lem:superadd} to treat the functional $\cE_Q[\Psi_N]$ (instead of $\|\Psi_N\|_{\dot H^{s,N}(Q)}^2$). To be precise, for any subset $B$ of $\{1,\ldots,N\}$ we denote by $\1_B$ the characteristic function of the set 
$$ 
	\Bigl\{ (\bx_1,\dots,\bx_N) \in (\R^d)^N : \bx_i\in Q \Leftrightarrow\, i \in B, \text{ for all } i\Bigr\}.
$$

For any $\Psi_N\in \cH^{s,N}_k(\R^d)$,  $\|\Psi_N\|_{L^2(\R^{dN})}=1$, by inserting the partition of unity 
\begin{align} \label{eq:pt-B}
	\1_{\R^{dN}} = \sum_{B} \1_B
\end{align}
we can write
\begin{align} \label{eq:local-EQ-lowerbound}
	\cE_Q[\Psi_N] &=\biggl\langle \Psi_N, \sum_{i=1}^N (-\Delta_{\bx_i})^s_{|Q} \Psi_N \biggr\rangle 
	= 
	\sum_{B} \biggl\langle \1_B \Psi_N, \sum_{i=1}^N (-\Delta_{\bx_i})^s_{|Q} \1_{B} \Psi_N \biggr\rangle\nonumber\\
	&=  
	\sum_{B} \int_{(\R^{d}\setminus Q)^{N-|B|}} \|\Psi_N(\,\cdot\,;\sx_{B^c})\|^2_{\dot H^{s,|B|}(Q)} \1_{ (\R^d\setminus Q)^d}d \sx_{B^c} \nonumber\\
	&\ge   
	\sum_{B} \int_{(\R^{d}\setminus Q)^{N-|B|}} E_{|B|}(Q)\|\Psi_N(\,\cdot\,;\sx_{B^c})\|^2_{L^2(Q^{|B|})} d \sx_{B^c}\nonumber\\
	&= 
	\sum_{B} E_{|B|}(Q) \langle \Psi_N, \1_B \Psi_N\rangle.
\end{align}
Here we have used the fact that $\1_B$ commutes with $(-\Delta_{\bx_i})^s_{|Q}$ and the shorthand notation
$$
	(\bx_1,\ldots,\bx_N)=(\sx_{B};\sx_{B^c}), \quad    \sx_{B}= (\bx_\ell)_{\ell \in B} \in (\R^{d})^{ |B|}.
$$
On the other hand, the partition of unity \eqref{eq:pt-B} implies that 
$$
	\sum_{B} \langle \Psi_N, \1_B \Psi_N\rangle = \langle \Psi_N, \Psi_N\rangle =1 
$$
and 
$$
	\sum_{B} |B|  \langle \Psi_N, \1_B \Psi_N\rangle = \sum_{B} \biggl \langle \1_B \Psi_N, \sum_{i=1}^N \1_Q(\bx_i) \1_B \Psi_N \biggr \rangle 
	= \biggl \langle \Psi_N, \sum_{i=1}^N \1_Q(\bx_i)  \Psi_N \biggr \rangle= \int_Q \varrho_{\Psi_N}. 
$$
Thus from \eqref{eq:local-EQ-lowerbound} and \eqref{eq:EnQ-lower-bound} we conclude that 
\begin{align*}
	\cE_Q[\Psi_N]	&\ge \frac{C}{|Q|^{2s/d}} \sum_{B} \Bigl[|B|-q\Bigr]_+ \langle \Psi_N, \1_B \Psi_N\rangle \\
	&\ge 
	\frac{C}{|Q|^{2s/d}} \Bigl[ \sum_{B} (|B|-q)  \langle \Psi_N, \1_B \Psi_N\rangle\Bigr]_+\\
	&= 
	\frac{C}{|Q|^{2s/d}} \Bigl[ \int_{Q} \varrho_{\Psi_N} - q \Bigr]_+
\end{align*}
by Jensen's inequality and the convexity of the function $t\mapsto [t]_+$.
\end{proof}

\section{Many-body Poincar\'e inequality} \label{sec:Poincare}

The crucial fact that the local energy $E_n(\Omega)$ in \eqref{eq:def-EnQ} eventually becomes positive
with increasing particle number is the content of the following Poincar\'e inequality:

\begin{theorem}[Poincar\'e inequality for functions vanishing on diagonals]\label{thm:Poincare}
	Fix an integer $k\geq 2$ and a bounded connected Lipschitz domain $\Omega\subset \R^d$. Assume that $2s>d(k-1)$. 
	For $N \in \N$ large enough ($N \ge \lceil s \rceil^d k$ is sufficient) there exists a positive constant $C$ depending only on $s, k, N, \Omega$ so that 
	\begin{equation}\label{eq:Poincare_ineq_thm}
		\|u\|_{\dot{H}^{s,N}(\Omega)}\geq C \|u\|_{L^2(\Omega^N)}
	\end{equation}
	for all $u\in C^\infty(\Omega^N)$ whose restriction to $\bDelta_k$ is zero.
\end{theorem}

Since Theorem \ref{thm:Poincare} is of independent interest, we state the result for more general domains although the result for cubes is sufficient for our application. 

\begin{proof}[Conclusion of Theorem  \ref{thm:main-general}] From Theorem \ref{thm:Poincare} and Lemma \ref{lem:sim-local-exclusion} we obtain the local exclusion bound \eqref{eq:loca-exclusion-vanishing}. Theorem \ref{thm:main-general} then immediately follows from the proof strategy in Section~\ref{sec:general-strategy}. 
\end{proof}

It remains to prove Theorem \ref{thm:Poincare}. The central fact used in the proof is that a function minimizing~\eqref{eq:Poincare_ineq_thm} must be a polynomial, and that if a polynomial vanishes on too many diagonals it must be zero.

\begin{lemma}[Low-degree polynomials vanishing on diagonals are trivial]\label{lem:Polynomial_vanishing} 
Given $d,k,S\in \N_1$ and $N \ge (S+1)^d k$. Let the $dN$-variable polynomial $f(\bx_1,\ldots,\bx_N)$, with $\bx_i\in \R^d$, satisfy 
\begin{itemize}
	\item $\deg_{\bx_j} f\le S$ for all $j \in \{1,\ldots,N\}$, 
	\item $f(\bx_1,\ldots,\bx_N)=0$ on $\bDelta_k$.
\end{itemize}
Then $f\equiv 0$. 
\end{lemma}

\begin{proof} The case $k=1$ ($\bDelta_1 = \R^d$) is trivial. We prove the other cases by induction. 

\bigskip

\noindent{\bf Step 1:} Consider $d=1$ and $k=2$.  Then $f(x_1,\ldots,x_N)=0$ if $x_i=x_j$ for some $i\ne j$. Consequently, when $x_2,\ldots,x_N$ are mutually different, the one variable polynomial $g(x_1)=f(x_1,x_2,\ldots,x_N)$ has $\deg g \le S$ but it has $N-1$ different roots $x_1=x_2,\ldots,x_1=x_N$.  Therefore, if 
$$N-1 > S$$
(which holds if $N \ge (S+1)k$) then $g(x_1)\equiv 0$. Thus 
$$
f(x_1,\ldots,x_N)=0
$$
for all $x_1,\ldots,x_N \in \R$ satisfying that $x_2,\ldots,x_N$ are mutually different. By continuity, we conclude that $f\equiv 0$. 

\bigskip

\noindent{\bf Step 2:} Consider $d=1$ and $k> 2$. Then $f(x_1,\ldots,x_N)=0$ if at least $k$ points $x_i$'s coincide. Then if $x_{k},\ldots,x_N$ are mutually different, the one-variable polynomial 
$$g(x_1)=f(x_1,\ldots,x_1,x_k,\ldots,x_N)$$
has $\deg g\le S(k-1)$ but it has $N-k+1$ different roots $x_1=x_k,\ldots,x_1=x_N$. Therefore, if 
$$
N-k+1 > S(k-1)
$$ 
(which holds if $N \ge (S+1)k $) then $g\equiv 0$. Thus
$$
f(x_1,\ldots,x_1,x_k,\ldots,x_N)=0
$$
if $x_k,\ldots,x_N$ are mutually different. By continuity, we conclude that
$$
f(x_1,\ldots,x_1,x_k,\ldots,x_N)=0
$$
for all $x_1,\ldots,x_N$. Similarly, by a renumbering, we can show that 
$$
f(x_1,x_2,\ldots,x_N)=0
$$
if at least $(k-1)$ points $x_i$'s coincide. By induction in $k$, we conclude that $f\equiv 0$. 

\bigskip

\noindent{\bf Step 3:} Now consider $d>1$ and $k\ge 2$.  Let us denote 
$$
\bx_i= (y_i,\bz_i) \in \R\times \R^{d-1}.  
$$
Take 
$$n= (S+1)k, \quad N\ge (S+1)^{d} k = (S+1)^{d-1}n.$$ 
Then for any $\bz \in \R^{d-1}$ and $\bx_{n+1},\ldots,\bx_N\in \R^d$, the polynomial 
$$
g(y_1,\ldots,y_n)= f((y_1,\bz), \ldots,(y_n,\bz), \bx_{n+1},\ldots,\bx_N)
$$
satisfies that $\deg_{y_i} g\le S$ and $g=0$ if (at least) $k$ points $y_i$'s coincide. By the result in the 1D case (with the choice $n=(S+1)k$) we conclude that $g \equiv 0$. Similarly, we obtain that 
$$
f(\bx_1,\ldots,\bx_N)= f((y_1,\bz_1),\ldots,(y_N,\bz_N))=0 
$$
if at least $n$ points $\bz_i$'s coincide. By induction in $d$ (i.e.\ using the induction hypothesis with $d-1$ and $k=n$, $N\ge (S+1)^{d-1}n$) we conclude that $f\equiv 0$. 
\end{proof}

We will also need the following technical lemma, which essentially states that if a multivariable function is a polynomial in each variable separately, then it is a multivariable polynomial.  The proof of this seemingly obvious fact is indeed non-trivial; see Carroll~\cite{Carroll-61} for an elegant proof in the two variables case. Here we provide an alternative proof for $n$ variables. 

\begin{lemma} \label{lem:multivariate-polynomial}
	Let $f(x_1,\ldots,x_n) \in L^1_{\rm loc}(\R^n)$ satisfy that for any $j=1,2,\ldots,n$ and for a.e. $(x_1,\ldots,x_{j-1},x_{j+1},\ldots,x_n) \in \R^{n-1}$ the mapping $x_j\mapsto f(x_1,\ldots,x_j,\ldots,x_n)$ is a polynomial
	of degree at most $M_j$. Then f is a polynomial of n variables $(x_1,\ldots,x_n)$
	of degree at most $M=\sum_{j=1}^n M_j$. 
\end{lemma}

From the proof below, it is clear that we can replace $\R^n$ by a subdomain (e.g. a cube). 

\begin{proof} {\bf Step 1.} We use the notation $\sx=(x_1,\ldots,x_n)\in \R^n$, $\alpha=(\alpha_1,\ldots,\alpha_n)\in \mathbb{N}_0^n$, and for every $j=1,2,\ldots,n$ we write 
$$
\sx=(x_j;\sx_{j}'), \quad \alpha=(\alpha_j;\alpha_j'). 
$$

By assumption, for a.e. $\sx_j'\in \R^{n-1}$, the mapping $x_j\mapsto f(x_j;\sx_j')$ is a polynomial of degree at most $M_j$. Therefore, for any $\alpha_j>M_j$, $D^{\alpha_j} f(\,\cdot\,; \sx_j')=0$ as distribution on $\R$, namely
\begin{equation} \label{eq:pol=0-1D}
	\int_{\R} f(x_j;\sx_j') D^{\alpha_j} h(x_j) \,d x_j =0, \quad \forall h\in C_c^\infty(\R). 
\end{equation}
Consequently, $D^{\alpha}f=0$ as distribution in $\R^n$ if $|\alpha|>M$. Indeed, since $|\alpha|>M$ we have $\alpha_j>M_j$ for some $j$, and hence for any test function $\varphi\in C_c^\infty(\R^n)$ using Fubini's theorem and \eqref{eq:pol=0-1D} we can write 
$$
\int_{\R^n} f D^\alpha \varphi \,d\sx = \int_{\R^{n-1}} \biggl[ \int_{\R} f(x_j;\sx_j') D^{\alpha_j} \Bigl( D^{\alpha_j'} \varphi(x_j;\sx_j') \Bigr) d x_j  \biggr] d \sx_j' =0. 
$$ 

\noindent 
{\bf Step 2.} Thus it remains to prove that if $D^{\alpha}f=0$ as distribution in $\R^n$ for any $|\alpha|>M$, then $f$ is a polynomial of $n$ variables. We prove this statement by induction in $M$. 

If $M=0$, then $D_{x_j}f=0$ as distribution for any $j=1,2,\ldots,n$, and hence $f$ is constant by~\cite[Theorem~6.1]{LiebLoss_book}. 

Now we prove the statement for $M \ge 1$ using the induction hypothesis for $M-1$. From
$$D^{\alpha}f=0, \quad \forall |\alpha|>M$$
we have for any $j=1,2,\ldots,n$, 
$$D^{\alpha} (D_{x_j}f)=0, \quad \forall |\alpha|>M-1.$$
Thus by the induction hypothesis for $M-1$, $D_{x_j}f$ is a polynomial of $n$ variables for any $j=1,2,\ldots,n$. Since $D_{x_j}f \in C(\R^n)$ for all $j=1,2,\ldots,n$, we obtain that $f \in C^1(\R^n)$ by~\cite[Theorem~6.10]{LiebLoss_book} and we have the formula~\cite[Theorem~6.9]{LiebLoss_book}
$$
f(\sx)=f(0)+ \int_0^1 \sx \cdot (\nabla f)(t\sx) \,d t, \quad \forall \sx\in \R^n. 
$$ 
The latter formula and the fact that $D_{x_j}f$ is a polynomial of $n$ variables for any $j=1,2,\ldots,n$ imply that $f$ is a polynomial of $n$ variables. This ends the proof. 
\end{proof}

\begin{proof}[Proof of Theorem~\ref{thm:Poincare}]
	We argue by contradiction. Assume that~\eqref{eq:Poincare_ineq_thm} is false, then there exists a sequence $u_n\in C^\infty(\Omega^N)$ satisfying $\|u_n\|_{L^2}=1,$ $u_n\bigl|_{\scriptbDelta_k}\equiv 0$, and
	\begin{equation}\label{eq:contradicting_seq}
		\|u_n\|_{\dot{H}^{s,N}(\Omega)}\to 0, \quad \mbox{as }n\to \infty.
	\end{equation}
	In particular, $u_n$ is bounded in the Sobolev space $H^{\nu}(\Omega^N)$ with $\nu=\min\{s, 1\}$. Indeed, for $d=1$ this follows from Lemma~\ref{lem:Equivalence_local_spaces} and Sobolev's embedding theorem. If $d\geq 2$ then $s>1$ and the claim follows from Sobolev's embedding theorem combined with that for any $\Omega$ the $\dot H^{1}(\Omega^N)$ and $\dot H^{1,N}(\Omega)$ seminorms are equivalent. By compactness of the embedding $H^\nu(\Omega^N) \subset L^2(\Omega^N)$, up to a subsequence, $u_n$ converges strongly to a function $P$ in $L^2(\Omega^N)$. Since $\|u_n\|_{L^2(\Omega^N)}=1$ we have that $\|P\|_{L^2(\Omega^N)}=1$.

	On the other hand, by Poincar\'e's inequality for $\dot{H}^s(\Omega)$ (combining~\cite[Theorem~8.11]{LiebLoss_book} and~\cite[Lemma~2.2]{Hurri_etal_13}) 
	\begin{align*}
		\|u_n\|_{\dot{H}^{s,N}(\Omega)}^2
		&=
		\sum_{j=1}^N \int_{\Omega^{N-1}}\|u_n(x_j; \sx')\|^2_{\dot{H}^s_{\bx_j}(\Omega)}d\sx'\\
		&\geq
		C\sum_{j=1}^N \int_{\Omega^{N}}\bigl|u_n(\sx)- P^{(n)}_j(\sx)\bigr|^2\,d\sx,
	\end{align*}
	where $P_j^{(n)}(\sx)$ is a polynomial in $\bx_j$ of degree $\leq \lceil s-1\rceil$. 
	In fact, the polynomial can be written explicitly as 
	\begin{equation}
		P_j^{(n)}(\sx) = \sum_{|\beta|\leq \lceil s-1\rceil} \bx_j^\beta \langle \varphi_\beta(\bx_j), u_n(\bx_j; \sx')\rangle_{L^2_{\bx_j}(\Omega)}
	\end{equation}
	for universal functions $\varphi_\beta\in C^\infty(\Omega)$. Since $u_n$ converges strongly in $L^2(\Omega^N)$, we can conclude that $P^{(n)}_j(\sx) \to P_j(\sx)$ strongly in $L^2(\Omega^N)$ and the limit is again a polynomial in $\bx_j$ of degree $\leq \lceil s-1\rceil$. The assumption \eqref{eq:contradicting_seq} allows us to identify the limiting functions and we find that
	\begin{equation*}
		P(\sx)=P_j(\sx)\mbox{ in } L^2(\Omega^N), \quad  \forall j. 
	\end{equation*}

	Thus the function $P(x)$ is a polynomial in each variable $\bx_j$ (of degree $\leq \lceil s-1\rceil$). 
	By Lemma~\ref{lem:multivariate-polynomial}, $P(x)$ is a multivariate polynomial whose degree in each $\bx_j$ is $\leq \lceil s-1\rceil$.

	We now want to use that $u_n=0$ on $\bDelta_k$ to prove that $P=0$ on $\bDelta_k$. Once this is done, then Lemma~\ref{lem:Polynomial_vanishing} implies that $P\equiv 0$ if $N \ge \lceil s \rceil^d k$. 
	This contradicts that $\|P\|_{L^2(\Omega^N)}=1$ and hence completes our proof. Note that if we can prove that $P\equiv 0$ in some open subset this is sufficient, in particular we can find some open cube $Q\subseteq \Omega$ and consider instead $u_n$ and $P$ restricted to $Q^N$.

	We consider the diagonal $\bx_1=\bx_2=\ldots = \bx_k$; the other cases are treated identically.
	By Lebesgue's differentiation theorem it suffices to prove that
	\begin{equation}\label{eq:Lebesgue_differentiation}
		\lim_{\delta\to 0}\frac{1}{\delta^{d(k-1)}}\int_{Q^{N-k+1}}\int_{\max_{j\leq k}|\bx_1-\bx_j|< \delta}|P(\bx_1, \ldots, \bx_k; \sx')| \,d\sx =0.
	\end{equation}
	By Fatou's lemma we have for any $\delta>0$ that
	\begin{equation}\label{eq:Fatou_diagonals}
	\begin{aligned}
		\int_{Q^{N-k+1}}&\int_{\max_{j\leq k}|\bx_1-\bx_j|< \delta}|P(\bx_1, \ldots, \bx_k; \sx')| \,d\sx\\
		&\leq 
		\lim_{n\to \infty}
		\int_{Q^{N-k+1}}\int_{\max_{j\leq k}|\bx_1-\bx_j|< \delta}|u_n(\bx_1, \ldots, \bx_k; \sx')| \,d\sx.
	\end{aligned}
	\end{equation}

	Since $u_n=0$ on $\bDelta_k$ it holds that
	\begin{equation}\label{eq:Holder_step1}
	\begin{aligned}
		\int_{Q^{N-k+1}}&\int_{\max_{j\leq k}|\bx_1-\bx_j|< \delta}|u_n(\bx_1, \ldots, \bx_k; \sx')| \,d\sx\\
		&=
		\int_{Q^{N-k+1}}\int_{\max_{j\leq k}|\bx_1-\bx_j|< \delta}|u_n(\bx_1, \ldots, \bx_k; \sx')-u_n(\bx_1, \ldots, \bx_1; \sx')| \,d\sx.
	\end{aligned}
	\end{equation}
	By Lemma~\ref{lem:Equivalence_local_spaces}, any $u\in L^2(Q^{l})$ with $\|u\|_{\dot{H}^{s,l}(Q)}<\infty$ satisfies that $u\in H^s(Q^{l})$ and moreover there is a constant $C$ depending only on $Q, l, s$ such that
	\begin{equation}\label{eq:Norm_equiv_Poincare}
		\|u\|_{H^s(Q^{l})}\leq C\bigl(\|u\|_{L^2(Q^l)}+\|u\|_{\dot{H}^{s,l}(Q)}\bigr).
	\end{equation}
	If $2s>dl$,
	by Sobolev's embedding theorem (see for instance~\cite[Theorem~8.2]{Hitchhikers}), 
	there is for any $\gamma\in \bigl(0, \min\bigl\{1, \frac{2s-dl}{2}\bigr\}\bigr)$
	a constant $C$ so that 
	\begin{equation*}
		\|u\|_{C^{0,\gamma}(Q^l)}\leq C \|u\|_{H^s(Q^l)}, \quad \mbox{for all }u\in H^s(Q^l).
	\end{equation*}
	By assumption $2s>d(k-1)$, and hence we can apply this result to the function
	\begin{equation*}
		(\bx_2, \ldots, \bx_k)\mapsto u_n(\bx_1, \ldots, \bx_k; \sx')
	\end{equation*} 
	(whose $\dot{H}^{s,k-1}(Q)$-seminorm is bounded for a.e.~$(\bx_1, \sx')$). Equation~\eqref{eq:Holder_step1} then implies that
	\begin{align*}
		\int_{Q^{N-k+1}}&\int_{\max_{j\leq k}|\bx_1-\bx_j|< \delta}|u_n(\bx_1, \ldots, \bx_k; \sx')| \,d\sx\\
		&\leq
		C\int_{Q^{N-k+1}}\int_{\max_{j\leq k}|\bx_1-\bx_j|< \delta} \|u_n(\bx_1, \sx''; \sx')\|_{H^s_{\sx''}(Q^{k-1})}|\sx''-\sx''_1|^\gamma d\sx'' d\bx_1 d\sx',
	\end{align*}
	where we set $\sx''=(\bx_2, \ldots, \bx_k)$ and $\sx''_1=(\bx_1, \ldots, \bx_1)$. Applying~\eqref{eq:Norm_equiv_Poincare} and H\"older's inequality yields
	\begin{align*}
		\int_{Q^{N-k+1}}&\int_{\max_{j\leq k}|\bx_1-\bx_j|< \delta}|u_n(\bx_1, \ldots, \bx_k; \sx')| \,d\sx
		\leq
		C\delta^{d(k-1)+\gamma}\bigl(\|u_n\|_{L^2(Q^N)}+\|u_n\|_{\dot{H}^{s,N}(Q)}\bigr).
	\end{align*}
	Since $\|u_n\|_{L^2(Q^N)}+\|u_n\|_{\dot{H}^{s,N}(Q)}\leq C$ and $\gamma>0$, we arrive at~\eqref{eq:Lebesgue_differentiation} which completes the proof of Theorem~\ref{thm:Poincare}.
\end{proof}

We finally note that the many-body nature of the wave functions is crucial for Theorem~\ref{thm:Poincare} to hold. The following example shows that  the requirement that the particle number $N$ is large, 
in fact typically strictly larger than $k$, is necessary. 

\begin{proposition}[Counterexample to the $k$-body case]\label{prop:counterex}
Theorem~\ref{thm:Poincare} cannot hold for $N < k$, or for $N = k$ if $s$ is integer and
$$
	\max\{d,2\}(k-1) < 2s < (d+k)(k-1).
$$
Replacing the condition $u\rvert_{\scriptbDelta_k}=0$ by the stronger condition
$$
	u \in \cH_{0,k}^{s,N}(\R^d) := \overline{\bigl\{ \Psi \in C_c^\infty(\R^{dN} \setminus \bDelta_k) \bigr\} }^{\, H^s(\R^{dN})},
$$
or
$$
	u \in \cH_{W,k}^{s,N}(\R^d) := \biggl\{  \Psi \in H^s(\R^{dN}):  \int_{\R^{dN}} W_{s,k} |\Psi|^2 < \infty \biggr\},
$$
with the $k$-particle generalization of $W_s$,
$$
	W_{s,k}(\sx) := \sum_{\substack{A \subseteq \{1,\ldots,N\}\\|A|=k}} \Bigl( \sum_{\substack{j,l\in A\\j<l}} |\bx_j-\bx_l|^2 \Bigr)^{-2s},
$$
does not help.
\end{proposition}

\begin{proof}
If $N < k$ there is no diagonal set $\bDelta_k$ and we may take the constant function
as a counterexample.
For $N = k$ we consider the polynomial
$$
	u(\bx_1,\ldots,\bx_k) := \prod_{1 \le j<l \le k} (x_{j,1} - y_{l,1}),
$$
for which, by the arithmetic mean-geometric mean inequality
and the triangle inequality,
$$
	|u(\sx)|^2 \lesssim \biggl( \sum_{j<l} |x_{j,1}-x_{l,1}|^2 \biggr)^{\binom{k}{2}}
	\le \biggl( \sum_{j<l} |\bx_j-\bx_l|^2 \biggr)^{\binom{k}{2}}
	\lesssim \biggl( \sum_{l\geq 2} |\bx_1-\bx_l|^2 \biggr)^{\binom{k}{2}}=: R^{2\binom{k}{2}},
$$
where $R \ge 0$ may serve as a radial coordinate on $\R^{d(k-1)}$ relative to $\bx_1$.
Hence, we have that
$$
	\int_{Q^k} W_{s,k} |u|^2
	\lesssim \int_Q \int_{Q^{k-1}} \frac{R^{2\binom{k}{2}}}{R^{2s}} d\bx_2 \ldots d\bx_k d\bx_1 
	\lesssim \int_0^C R^{k(k-1)-2s+d(k-1)-1} dR < \infty,
$$
if $d(k-1) < 2s < (d+k)(k-1)$. Thus (analogously to Lemma~\ref{lem:inclusion_of_spaces_W}, and by extension)
$$
	u \in \cH_{W,k}^{s,N}(\R^d)
	\subseteq \cH_{0,k}^{s,N}(\R^d).
$$
On the other hand
$$
	\|u\|_{\dot{H}^{s,k}(\Omega)}^2 
	= \sum_{j=1}^k \sum_{|\alpha|=m} \frac{m!}{\alpha!} \int_{Q^k} |D^\alpha_{\bx_j} u|^2
	= 0,
$$
if $s=m > k-1$.
\end{proof}

A particular case included in the above is $d=3$, $s=2$, $k=2$, 
with the function $u(\bx,\by) := x_1 - y_1$.


\appendix
\section{Equivalence of Sobolev spaces}\label{app:norms}

In this appendix we discuss the $N$-particle space
\begin{equation*}
	H^{s,N}(\Omega) := \bigl\{u \in L^2(\Omega^N): \|u\|_{\dot{H}^{s,N}(\Omega)}<\infty\bigr\}
\end{equation*}
and its relation to the standard Sobolev space $H^s(\Omega^N)$.

If $\Omega=\R^d$ the equivalence of the seminorms (and consequently the spaces)
\begin{align}\label{eq:Norm_equiv}
	c_{s,N}\|u\|_{\dot{H}^s(\R^{dN})}\leq \|u\|_{\dot{H}^{s,N}(\R^{d})}\leq C_{s,N}\|u\|_{\dot{H}^s(\R^{dN})}
\end{align}
can be seen via the Fourier transform. 
However, the constants in the equivalence depend on $N$ and $s$. 
In particular, if $s\neq 1$ the equivalence degenerates as $N$ tends to infinity; either $c_{s,N}\to 0$ or $C_{s,N}\to \infty$. Specifically, the sharp constants in~\eqref{eq:Norm_equiv} are given by
\begin{align*}
	c_{s,N}=\min\{1, N^{(1-s)/2}\}
	\quad \mbox{and} \quad
	C_{s,N}= \max\{1, N^{(1-s)/2}\}.
\end{align*}
Thus it is a slightly subtle question of what happens to these spaces in the many-body limit. 
An even more subtle question is what happens to the local versions of these spaces, i.e.\ when $\R^d$ is replaced by $\Omega \subsetneq \R^d$. 
For us, the following equivalence of the spaces in the case of cubes will suffice:

\begin{lemma}\label{lem:Equivalence_local_spaces}
	Let $u\in L^2(Q^N)$, $Q=[0, 1]^{d}$. There exist positive constants $c, C$ depending only on $d, s, N$ so that
	\begin{equation*}
		c\bigl(\|u\|_{L^2(Q^N)}+\|u\|_{\dot{H}^{s,N}(Q)}\bigr)\leq \|u\|_{H^s(Q^N)}\leq C \bigl(\|u\|_{L^2(Q^N)}+\|u\|_{\dot{H}^{s,N}(Q)}\bigr).
	\end{equation*}
\end{lemma}

Lemma~\ref{lem:Equivalence_local_spaces} is an immediate consequence of the equivalence 
\eqref{eq:Norm_equiv}
of the two seminorms on~$\R^{dN}$ and the following extension lemma:
\begin{lemma}\label{lem:Extension_on_cubes}
	Let $u\in L^2(Q^N)$, $Q=[0, 1]^{d},$ and assume that $\|u\|_{L^2(Q^N)}+\|u\|_{\dot{H}^{s,N}(Q)}<\infty$. There exists a function $\tilde u\in L^2(\R^{dN})$ with compact support satisfying
	\begin{equation*}
		\tilde u\bigl|_{Q^N}=u, 
		\quad \mbox{and} \quad
		\|\tilde u\|_{L^2(\R^{dN})}+\|\tilde u\|_{\dot{H}^{s,N}(\R^{d})}\leq C \bigl(\|u\|_{L^2(Q^N)}+\|u\|_{\dot{H}^{s,N}(Q)}\bigr),
	\end{equation*}
	where $C$ is a constant depending only on $s$, $d$ and $N$.
\end{lemma}

\begin{proof}
	We shall prove the lemma by using higher-order reflection through one side of the hypercube $Q$ at a time.
	To this end we recall that if $v\in C^n([0, 1])$, for some $n\geq 0$, we can construct an explicit extension $\tilde v\in C^n((-\infty, 1])$ satisfying $\tilde v(x)=0$ when $x<-\delta$. Namely, set
	\begin{equation*}
		\tilde v(x)=
		\begin{cases}
			v(x), \quad &\mbox{if } x\in [0, 1],\\
			\varphi(x)\sum_{j=1}^{n+1} \lambda_j v(-x/j), \quad &\mbox{if } x<0,
		\end{cases}
	\end{equation*}
	where $\varphi \in C^\infty((-\infty, 0])$ such that $\varphi(x)\equiv 0$ for $x<-\delta$ and $\varphi(x)\equiv 1$ in $[-\delta/2, 0]$. What remains is to verify that we can choose the $\lambda_j$'s so that $\tilde v\in C^n$. But if we differentiate $\tilde v$ for $x$ away from zero we see that the system of equations that we need the $\lambda_j$ to satisfy to get continuity of the derivatives across $x=0$ is 
	\begin{equation*}
		\Bigl[(-j)^{1-i}\Bigr]_{i,j=1}^{n+1}\left( \begin{matrix}
			\lambda_1 \\ \vdots \\ \lambda_{n+1}
		\end{matrix}\right) =\left( \begin{matrix}
			1\\ \vdots \\ 1
		\end{matrix}\right).
	\end{equation*}
	But the determinant of this matrix is non-zero (it is a Vandermonde matrix) and hence there exists a unique solution $(\lambda_1, \ldots, \lambda_{n+1})$.

	We shall now prove that we can use this one-dimensional extension repeatedly to construct an extension of $u$ to $\R^{dN}$. The idea is to use the one-dimensional result one coordinate at a time and show that the new function in each step has the quantity corresponding to the $\dot{H}^{s,N}$-seminorm controlled by that of $u$. 

	Without loss we can assume that $u\in C^n(Q^N)$ (the construction is stable under approximation), where we take $n=\lceil s\rceil$. Consider $u(x_1; \sx')$, $x_1\in [0, 1]$ and $\sx'\in [0, 1]^{dN-1}$. And apply the above lemma for each fixed $\sx'$, that is, we define $v_1$ by
	\begin{equation*}
		v_1(x_1; \sx')=
		\begin{cases}
			u(x_1; \sx'), \quad &\mbox{if } x_1\in [0, 1],\\
			\varphi(x_1)\sum_{j=1}^{n+1} \lambda_j u(-x_1/j; \sx'), \quad &\mbox{if } x_1\in [-1, 0).
		\end{cases}
	\end{equation*}
	It is a simple calculation to use Sobolev's embedding theorem to prove that we can bound the $L^p$-norm of $l$-th order derivatives of $v_1$ by the corresponding one for $u$ if $l\leq n$. We need to prove that also the fractional order seminorm is preserved. That is, we wish to show that, with $s=m+\sigma$ and $Q'=[-1, 1]\times [0, 1]^{d-1}$,
	\begin{equation}\label{eq:NormBound_reflection}
	\begin{aligned}
		&\int_{Q^{N-1}} \iint_{Q'\times Q'} \frac{|D^\alpha_{\bx_1} v_1(\bx_1;\sx')-D^\alpha_{\by_1} v_1(\by_1; \sx')|^2}{|\bx_1-\by_1|^{d+2\sigma}}d\bx_1 d\by_1 d\sx'\\[5pt]
		&\quad+
		\sum_{i=2}^N\int_{Q'\times Q^{N-2}}\iint_{Q\times Q} \frac{|D^\alpha_{\bx_i} v_1(\bx_i;\sx')-D^\alpha_{\by_i} v_1(\by_i; \sx')|^2}{|\bx_i-\by_i|^{d+2\sigma}}d\bx_i d\by_i d\sx'\\
		&\leq 
		C \bigl(\|u\|_{\dot{H}^{s,N}(Q)}^2+ \|u\|_{L^2(Q^N)}^2\bigr),
	\end{aligned}
	\end{equation}
	for all multi-indices $|\alpha|=m$. If we can prove this inequality, then by repeating the procedure to extend $v_1$ to $x_1>1$ the same proof gives that we can bound the corresponding $\dot{H}^{s,N}$ quantity in terms of that of $v_1$, and hence $u$. By repeating the procedure for each coordinate at a time we, after $2dN$ reflections, find a function $\tilde u\in L^2(\R^{dN})$ satisfying the claims of the lemma. Thus all that remains is to prove~\eqref{eq:NormBound_reflection}.

	\medskip

	We start with the first term which is also the most difficult:
	\begin{equation}\label{eq:Integral_split_extension}
	\begin{aligned}
		&\int_{Q^{N-1}} \iint_{Q'\times Q'} \frac{|D^\alpha_{\bx_1} v_1(\bx_1;\sx')-D^\alpha_{\by_1} v_1(\by_1; \sx')|^2}{|\bx_1-\by_1|^{d+2\sigma}}d\bx_1 d\by_1 d\sx'\\[5pt]
		&\quad=
		\int_{Q^{N-1}} \biggl[
		\iint_{Q\times Q} \frac{|D^\alpha_{\bx_1} u(\bx_1;\sx')-D^\alpha_{\by_1} u(\by_1; \sx')|^2}{|\bx_1-\by_1|^{d+2\sigma}}d\bx_1 d\by_1\\
		&\quad +
		2\iint_{Q\times (Q'\setminus Q)} \frac{|D^\alpha_{\bx_1} u(\bx_1;\sx')-\sum_{j=1}^{n+1}\lambda_j D^\alpha_{\by_1} (\varphi(y_{1,1}) u(-y_{1,1}/j, \by_1'; \sx'))|^2}{|\bx_1-\by_1|^{d+2\sigma}}d\bx_1 d\by_1\\
		&\quad +
		\iint_{(Q'\setminus Q)\times (Q'\setminus Q)} |\bx_1-\by_1|^{-d-2\sigma}\Bigl|\sum_{j=1}^{n+1}\lambda_j \Bigl(D^\alpha_{\bx_1} \bigl[\varphi(x_{1,1}) u(-x_{1,1}/j, \bx_1';\sx')\bigr]\\
		&\quad 
		-D^\alpha_{\by_1} \bigl[\varphi(y_{1,1}) u(-y_{1,1}/j, \by_1';\sx')\bigr]\Bigr)\Bigr|^2 d\bx_1 d\by_1 \biggr]d\sx'.
	\end{aligned}
	\end{equation}
	Clearly the integral over $Q\times Q$ is bounded by $\|u\|_{\dot{H}^{s,N}(Q)}$. We treat the two remaining terms separately.
	In order to bound the integral over $Q\times (Q'\setminus Q)$ we write 
	\begin{align*}
		Q_1 &=\{\bx\in Q'\setminus Q: x_1>-\delta/2\},\\
		Q_2 &=\{\bx\in Q'\setminus Q: x_1\leq-\delta/2\}.
	\end{align*}
	Thus we can bound the second integral in~\eqref{eq:Integral_split_extension} as follows:
	\begin{align*}
		&\int_{Q^{N-1}}\iint_{Q\times (Q'\setminus Q)} \frac{|D^\alpha_{\bx_1} u(\bx_1;\sx')-\sum_{j=1}^{n+1}\lambda_j D^\alpha_{\by_1} (\varphi(y_{1,1}) u(-y_{1,1}/j, \by_1'; \sx'))|^2}{|\bx_1-\by_1|^{d+2\sigma}}d\bx_1 d\by_1 d\sx'\\
		&\quad=
		\int_{Q^{N-1}}\iint_{Q\times Q_1} \frac{|D^\alpha_{\bx_1} u(\bx_1;\sx')-\sum_{j=1}^{n+1}\lambda_j D^\alpha_{\by_1} (u(-y_{1,1}/j, \by_1'; \sx'))|^2}{|\bx_1-\by_1|^{d+2\sigma}}d\bx_1 d\by_1 d\sx'\\
		&\quad+
		\int_{Q^{N-1}}\iint_{Q\times Q_2} \frac{|D^\alpha_{\bx_1} u(\bx_1;\sx')-\sum_{j=1}^{n+1}\lambda_j D^\alpha_{\by_1} (\varphi(y_{1,1}) u(-y_{1,1}/j, \by_1'; \sx'))|^2}{|\bx_1-\by_1|^{d+2\sigma}}d\bx_1 d\by_1 d\sx'\\
		&\quad\leq
		\int_{Q^{N-1}}\iint_{Q\times Q_1} \frac{|D^\alpha_{\bx_1} u(\bx_1;\sx')-\sum_{j=1}^{n+1}\lambda_j(-j)^{-\alpha_1} D^\alpha_{\by_1}u(-y_{1,1}/j, \by_1'; \sx')|^2}{|\bx_1-\by_1|^{d+2\sigma}}d\bx_1 d\by_1 d\sx'\\
		&\quad+
		\frac{C}{\delta^{d+2\sigma}}\int_{Q^{N-1}}\iint_{Q\times Q_2} |D^\alpha_{\bx_1} u(\bx_1;\sx')-\sum_{j=1}^{n+1}\lambda_j D^\alpha_{\by_1} (\varphi(y_{1,1}) u(-y_{1,1}/j, \by_1'; \sx'))|^2 d\bx_1 d\by_1 d\sx'.
	\end{align*}
	Using the triangle inequality and Sobolev's embedding theorem one finds that the second term is $\lesssim \|u\|_{L^2(Q^N)}^2+\|u\|_{\dot{H}^{s,N}(Q)}^2$. 
	Since $\sum_{j}\lambda_j (-j)^{-\alpha_1} = 1$ 
	for any $\alpha_1\leq m+1$, one obtains for the first integral 
	\begin{align*}
		&\int_{Q^{N-1}}\iint_{Q\times Q_1} \frac{|D^\alpha_{\bx_1} u(\bx_1;\sx')-\sum_{j=1}^{n+1}\lambda_j(-j)^{-\alpha_1} D^\alpha_{\by_1}u(-y_{1,1}/j, \by_1'; \sx')|^2}{|\bx_1-\by_1|^{d+2\sigma}}d\bx_1 d\by_1 d\sx'\\
		&\quad =
		\int_{Q^{N-1}}\iint_{Q\times Q_1} \frac{|\sum_{j=1}^{n+1}\lambda_j(-j)^{-\alpha_1}\bigl(D^\alpha_{\bx_1} u(\bx_1;\sx')-D^\alpha_{\by_1}u(-y_{1,1}/j, \by_1'; \sx')\bigr)|^2}{|\bx_1-\by_1|^{d+2\sigma}}d\bx_1 d\by_1 d\sx'\\
		&\quad \leq
		C\sum_{j=1}^{n+1}\int_{Q^{N-1}}\iint_{Q\times Q_1} \frac{|D^\alpha_{\bx_1} u(\bx_1;\sx')-D^\alpha_{\by_1}u(-y_{1,1}/j, \by_1'; \sx')|^2}{|\bx_1-\by_1|^{d+2\sigma}}d\bx_1 d\by_1 d\sx'\\
		&\quad \leq
		C\sum_{j=1}^{n+1}\int_{Q^{N-1}}\iint_{Q\times Q} \frac{|D^\alpha_{\bx_1} u(\bx_1;\sx')-D^\alpha_{\by_1}u(\by_1; \sx')|^2}{(|\bx_1'-\by_1'|^2+(x_{1,1}+jy_{1,1})^2)^{d/2+\sigma}} d\bx_1 d\by_1 d\sx'\\
		&\quad\leq C \|u\|_{\dot{H}^{s,N}(Q)}^2.
	\end{align*}
	In the last step we used the inequality $(x+jy)^2\geq (x-y)^2$ for $x,y\geq 0$ and $j\geq 1$.

	\medskip

	For the last integral in~\eqref{eq:Integral_split_extension} we have
	\begin{align*}
		&\int_{Q^{N-1}}
		\iint_{(Q'\setminus Q)\times (Q'\setminus Q)} |\bx_1-\by_1|^{-d-2\sigma}\Bigl|\sum_{j=1}^{n+1}\lambda_j \Bigl(D^\alpha_{\bx_1} \bigl[\varphi(x_{1,1}) u(-x_{1,1}/j, \bx_1';\sx')\bigr]\\
		&\qquad 
		-D^\alpha_{\by_1} \bigl[\varphi(y_{1,1}) u(-y_{1,1}/j, \by_1';\sx')\bigr]\Bigr)\Bigr|^2 d\bx_1 d\by_1 d\sx'\\
		&\quad=
		\int_{Q^{N-1}}\iint_{(Q'\setminus Q)\times (Q'\setminus Q)} |\bx_1-\by_1|^{-d-2\sigma}
		\Bigl|\sum_{j=1}^{n+1}\sum_{\gamma+\beta=\alpha_1}\lambda_j(-j)^{-\beta}\\
		&\qquad \times \Bigl( \varphi^{(\gamma)}(x_{1,1}) D_{\bx_1}^{\alpha'}u(-x_{1,1}/j, \bx_1';\sx')
		-\varphi^{(\gamma)}(y_{1,1}) D_{\by_1}^{\alpha'}u(-y_{1,1}/j, \by_1';\sx')\Bigr)\Bigr|^2 d\bx_1 d\by_1 d\sx',
	\end{align*}
	where we set $\alpha'$ as the multi-index $\alpha$ but with $\alpha_1$ exchanged for $\beta$. By the triangle inequality and the fact that $\sum_{j}\lambda_j(-j)^{-\beta}=1$ the integral is smaller than
	\begin{align*}
		&C\sum_{j=1}^{n+1}\int_{Q^{N-1}}\iint_{(Q'\setminus Q)\times (Q'\setminus Q)} |\bx_1-\by_1|^{-d-2\sigma}\Bigl|\sum_{\gamma+\beta=\alpha_1}\Bigl(\varphi^{(\gamma)}(x_{1,1}) D_{\bx_1}^{\alpha'}u(-x_{1,1}/j, \bx_1';\sx')\\
		&\qquad 
		-\varphi^{(\gamma)}(y_{1,1}) D_{\by_1}^{\alpha'}u(-y_{1,1}/j, \by_1';\sx')\Bigr)\Bigr|^2 d\bx_1 d\by_1 d\sx'\\
		&\quad\leq 
		C\sum_{j=1}^{n+1}\int_{Q^{N-1}}\iint_{Q\times Q} (|\bx_1'-\by_1'|^2+j^2(x_{1,1}-y_{1,1})^2)^{-d/2-\sigma}\\
		&\qquad \times 
		\Bigl|\sum_{\gamma+\beta=\alpha_1}\Bigl(\varphi^{(\gamma)}(-j x_{1,1}) D_{\bx_1}^{\alpha'}u(\bx_1; \sx')
		-\varphi^{(\gamma)}(- j y_{1,1}) D_{\by_1}^{\alpha'}u(\by_1; \sx')\Bigr)\Bigr|^2 d\bx_1 d\by_1 d\sx'\\
		&\quad\leq 
		C\sum_{j=1}^{n+1}\int_{Q^{N-1}}\iint_{Q\times Q} \frac{\bigl|D^\alpha_{\bx_1}\bigl[\varphi(-j x_{1,1})u(\bx_1; \sx')\bigr]
		-D^\alpha_{\by_1}\bigl[\varphi(- j y_{1,1}) u(\by_1; \sx')\bigr]\bigr|^2}{|\bx_1-\by_1|^{d+2\sigma}} d\bx_1 d\by_1 d\sx'\\
		&\quad \leq
		C \|u\|_{\dot{H}^{s,N}(Q)}^2,
	\end{align*}
	where we used that $\|\psi u\|_{\dot{H}^s(Q)}\leq C_\psi \|u\|_{\dot{H}^s(Q)}$ for any $\psi\in C^\infty(Q).$

	\medskip

	To show that the remaining terms in~\eqref{eq:NormBound_reflection} are $\lesssim \|u\|_{L^2}^2+\|u\|_{\dot{H}^{s,N}}^2$ one can proceed in an almost identical manner. The main difference is that in these terms the differentiation is with respect other variables than the variable in which the extension has been made, and the splitting of the integrals is slightly different. However, in the end this only simplifies each step of the proof.
\end{proof}


\section{Spaces of contact interaction}\label{app:subspaces}

We consider in the following only 2-particle diagonals $\bDelta$, for simplicity, 
however analogous statements can be made for the case of $k$-particle diagonals.

Define for $N\ge 2$ the restricted $N$-particle spaces
\begin{align*}
	\cH^{s,N}_W(\R^d) &:= \biggl\{  \Psi \in H^s(\R^{dN}):  \int_{\R^{dN}} W_s |\Psi|^2 < \infty \biggr\}, \\
	\cH^{s,N}_0(\R^d) &:= \overline{\bigl\{ \Psi \in C_c^\infty(\R^{dN} \setminus \bDelta) \bigr\} }^{\, H^s(\R^{dN})}, \\
	\cH^{s,N}(\R^d) &:= \overline{\bigl\{ \Psi \in C_c^\infty(\R^{dN}): \Psi\rvert_{\scriptbDelta}=0 \bigr\} }^{\, H^s(\R^{dN})}.
\end{align*}
Then we have for all $s>0$ the chain of inclusions
$$
	\cH^{s,N}_W(\R^d) \subseteq \cH^{s,N}_0(\R^d) \subseteq \cH^{s,N}(\R^d) \subseteq H^s(\R^{dN}).
$$
The latter two inclusions are trivial while the first one will be proved below.
Moreover, for $2s < d$ all four spaces are equal by the 
Hardy--Rellich inequality (see e.g.~\cite{Yafaev-99}):
$$
	\int_{\R^{dN}} |\bx_1-\bx_2|^{-2s}|\Psi(\bx_1;\sx')|^2 \,d\bx_1 d\sx'
	\le C\int_{\R^{d(N-1)}} \|\Psi\|_{\dot{H}^s_{\bx_1}(\R^d)}^2 d\sx'
	\le C\|\Psi\|_{H^s(\R^{dN})}^2.
$$ 
In the critical case $2s = d$ we still have 
$\cH_0^{s,N}(\R^d) = \cH^{s,N}(\R^d) = H^s(\R^{dN})$,
as is also shown below, but a strict inclusion
$\cH_W^{s,N}(\R^d) \subsetneq \cH_0^{s,N}(\R^d)$, as illustrated by
$$
	\Psi(\sx) = e^{-|\sx|^2}
$$
which is in $H^s(\R^{dN})$ but not in $\cH_W^{s,N}(\R^d)$ 
due to the non-integrability of $W_s$.
For $2s > d$ and $s-d/2 \notin \Z$
it again holds by the Hardy--Rellich inequality that 
$\cH_W^{s,N}(\R^d) = \cH_0^{s,N}(\R^d)$,
while not necessarily $\cH_0^{s,N}(\R^d) = \cH^{s,N}(\R^d)$, as with the example
$$
	\Psi(x_1,x_2) = (x_1-x_2)e^{-|\sx|^2}
$$
which is in $\cH^{s,2}(\R^d)$ but not in $\cH_W^{s,2}(\R^d)$ for $s=2$ and $d=1$.

Let $\chi_\eps^{(*)}(\sx) := \prod_{1 \le j < k \le N} \varphi_\eps^{(*)}(\bx_j-\bx_k)$
where $\varphi_\eps(\bx) = \varphi(|\bx|/\eps)$
and $\varphi_\eps^*(\bx) = \varphi^*(\eps\ln|\bx|)$.
We take $\varphi^{(*)}$ as smooth functions from $\R$ to $[0,1]$ such that
$\varphi(x) = 0$ for $x \le 1$,
$\varphi(x) = 1$ for $x \ge 2$, and
$\varphi^*(x) = 0$ for $x \le -2$,
$\varphi^*(x) = 1$ for $x \ge -1$.

\begin{lemma}\label{lem:norm_of_cutoff}
	Let $\Omega \subset \R^d$ be open and bounded. For all $s=m +\sigma > 0$, $d \ge 1$ and $N \ge 1$ it holds as $\eps \to 0$ that
	$$
		\| \chi_\eps \|_{\dot{H}^{s,N}(\Omega)} \le C \eps^{d/2-s},
		\qquad
		\| D^\alpha \chi_\eps \|_{\dot{H}^{\sigma,N}(\Omega)} \le C \eps^{d/2-|\alpha|-\sigma},
	$$
	while for $2s=d$
	$$
		\| \chi_\eps^* \|_{\dot{H}^{s,N}(\Omega)} \le C \eps^{1/2},
		\qquad
		\| D^\alpha \chi_\eps^* \|_{\dot{H}^{\sigma,N}(\Omega)} \le C \eps^{1/2}
	$$
	for $|\alpha| \le d/2-\sigma$.
\end{lemma}
\begin{proof}
	For $\alpha \neq 0$ there are in $D^\alpha_{\bx_j} \chi_\eps$ a total of 
	$|\alpha|$ derivatives of functions $\varphi_\eps(\bx_j-\bx_k)$,
	$k \neq j$, and remaining factors involving the other particles.
	These factors are uniformly bounded while each derivative yields an additional factor $1/\eps$,
	while reducing the support in $\bx_j$ to $B_{2\eps}(\bx_k) \setminus B_{\eps}(\bx_k)$.
	Furthermore, we thus have
	$$
		|D^\alpha \varphi_\eps(\bx)| 
		\le C \eps^{-|\alpha|} \1_{B_{2\eps}(0) \setminus B_{\eps}(0)},
	$$
	$$
		|D^\alpha \varphi_\eps(\bx) - D^\alpha \varphi_\eps(\by)| 
		\le C \eps^{-|\alpha|-1} |\bx-\by| \1_{\bx,\by \in B_{2\eps}(0) \setminus B_{\eps}(0)},
	$$
	and for
	$B(j,\eps) = \cup_{k \neq j} B_{2\eps}(\bx_k) \setminus \cup_{k \neq j} B_{\eps}(\bx_k)$,
	$$
		|\chi_\eps(\bx_j;\sx') - \chi_\eps(\by_j;\sx')| 
		\le C \eps^{-1}|\bx_j-\by_j| 
		\1_{\bx_j,\by_j \in B(j,\eps)},
	$$
	and
	$$
		|D^\alpha \chi_\eps(\bx_j;\sx') - D^\alpha \chi_\eps(\by_j;\sx')| 
		\le C \eps^{-|\alpha|-1}|\bx_j-\by_j| 
		\1_{\bx_j,\by_j \in B(j,\eps)}.
	$$
	Hence, 
	$\|D^\alpha \chi_\eps\|_{L^2_{\bx_j}(\Omega)}^2 \lesssim \eps^{-2|\alpha|+d}$,
	and for any $0<\sigma<1$
	\begin{align*}
		\|D^\alpha \chi_\eps\|_{\dot{H}^\sigma_{\bx_j}(\Omega)}^2 
		&= \iint_{\Omega \times \Omega} \dfrac{|D^\alpha\chi_\eps(\bx;\sx') - D^\alpha\chi_\eps(\by;\sx')|^2}{|\bx-\by|^{d+2\sigma}}\,d\bx d\by \\
		&\lesssim \eps^{-2|\alpha|-2} \sum_{k \neq j} \iint_{B_{2\eps}(\bx_k) \times B_{2\eps}(\bx_k)} 
			|\bx - \by|^{-d-2\sigma+2} d\bx d\by
		\lesssim \eps^{-2|\alpha|-2\sigma+d},
	\end{align*}
	so that $\|\chi_\eps\|_{\dot{H}^{s,N}(\Omega)}^2 \lesssim \eps^{-2s+d}$.
		
	Similarly, for $\chi_\eps^*$ we consider 
	$B(j,\eps) = \cup_{k \neq j} B_{e^{-1/\eps}}(\bx_k) \setminus \cup_{k \neq j} B_{e^{-2/\eps}}(\bx_k)$
	and
	\begin{equation}\label{eq:phi_star_bound}
		|D^\alpha\varphi_\eps^*(\bx)| 
		= |D^\alpha_\bx \varphi^*(\eps\ln|\bx|)| 
		\le C \eps |\bx|^{-|\alpha|} \1_{B_{e^{-1/\eps}}(0) \setminus B_{e^{-2/\eps}}(0)}.
	\end{equation}
	In $\chi_\eps^*$ this could involve different points $\bx_k$
	but the worst case is if they are the same,
	$$
		\|D^\alpha \chi_\eps^*\|_{L^2_{\bx_j}(\Omega)}^2
		\lesssim \eps^{2} \int_{B(j,\eps)} |\bx_j - \bx_k|^{-2|\alpha|} \,d\bx_j
		\lesssim \begin{cases} \eps^{2} & \text{for}\ 0<2|\alpha|<d,  \\ 
			\eps^{2} \int_{-2\eps^{-1}}^{-\eps^{-1}} ds = \eps & \text{for}\ 2|\alpha|=d. \end{cases} 
	$$
	This covers the even-dimensional critical case $d=2m$, $m \in \N_1$.
	
	In the odd-dimensional critical case $d=2m+2\sigma$, $\sigma = 1/2$,
	we observe that
	\begin{align*}
		\|D^\alpha \chi_\eps\|_{\dot{H}^\sigma_{\bx_j}(\Omega)}^2 
		&= \iint_{\Omega \times \Omega} \dfrac{|D^\alpha\chi_\eps(\bx;\sx') - D^\alpha\chi_\eps(\by;\sx')|^2}{|\bx-\by|^{d+1}}\,d\bx d\by \\
		&\lesssim \eps^{-2|\alpha|-2} \sum_{k \neq j} \iint_{B_{2\eps}(\bx_k) \times B_{2\eps}(\bx_k)} 
			|\bx - \by|^{-d+1} d\bx d\by
		\lesssim \eps^{-2|\alpha|-1+d},
	\end{align*}
	which is not enough for $2|\alpha|=d-1$.
	Instead we shall use $\chi_\eps^*$.

	For the case $2|\alpha|=d-1$ things are a bit less straightforward. We start with the case $d=1$ which is the easiest. Here our approach differs slightly due to the fact that in this case $|\alpha|=0$.

	Let $U_1=\cap_{k\neq j} B_{e^{-1/\eps}}(\bx_k)^c$, $U_2= \cup_{k\neq j} B_{e^{-2/\eps}}(\bx_k)$ and $U= \Omega\setminus(U_1\cup U_2)$.

	We estimate the seminorm $\|\chi_\eps^*\|_{\dot H^s_{\bx_j}(\Omega)}$. By construction of $\chi_\eps^{*}$ we have that
	\begin{equation*}
		|\chi_\eps^*(\bx; \sx')-\chi_\eps^*(\by; \sx')|\leq 1, \quad \forall \bx, \by \in \Omega.
	\end{equation*}
	Moreover, the difference is zero whenever $(\bx, \by) \in U_1^2 \cup U_2^2$.

	For $\bx$ and $\by$ close we need to estimate this quantity more precisely. By Taylor's theorem we can estimate
	\begin{align*}
		|\chi_\eps^*&(\bx; \sx') - \chi_\eps^*(\by; \sx')|\\
		&=
		\biggl|
		\int_0^1 \sum_{k\neq j}^{N}\Bigl(\prod_{i\notin\{k, j\}} \varphi^{*}_\eps(\bx-\bx_i)\Bigr)(\varphi^*)'(\eps(\ln|\bx-\bx_k|+t(\ln|\by-\bx_k|-\ln|\bx-\bx_k|)))\\
		&\qquad\times \eps(\ln|\by-\bx_k|-\ln|\bx-\bx_k|)\,dt\biggr|\\
		&\leq
		C \eps \sum_{k\neq j}\bigl|{\ln|\bx-\bx_k|-\ln|\by-\bx_k|}\bigr|.
	\end{align*}

	By symmetry in $\bx, \by$ we find
	\begin{equation}\label{eq:cc_integral_1D}
	\begin{aligned}
		\|\chi_\eps^*\|_{\dot H^s_{\bx_j}(\Omega)}^2
		&=
		\iint_{\Omega\times\Omega} \frac{|\chi_\eps^*(\bx; \sx')-\chi_\eps^*(\by; \sx')|^2}{|\bx-\by|^2}d\bx d\by\\
		&\leq
		2\iint_{U\times \Omega}\frac{|\chi_\eps^*(\bx; \sx')-\chi_\eps^*(\by; \sx')|^2}{|\bx-\by|^2}d\bx d\by
		+
		2\iint_{U_1\times U_2}\frac{1}{|\bx-\by|^2}d\bx d\by.
	\end{aligned}
	\end{equation}
	
	The latter term is fairly easy to estimate:
	\begin{equation*}
		\iint_{U_1\times U_2} \frac{1}{|\bx-\by|^2}d\bx d\by
		\leq
		2|U_1| e^{-2/\eps}\int_{-e^{-2/\eps}}^{e^{-2/\eps}}\frac{1}{(e^{-1/\eps}+r)^2}dr \leq C e^{-2/\eps}.
	\end{equation*}

	We return to the remaining term of~\eqref{eq:cc_integral_1D}:
	\begin{align*}
		\iint_{U\times \Omega}\frac{|\chi_\eps^*(\bx; \sx')-\chi_\eps^*(\by; \sx')|^2}{|\bx-\by|^2}d\bx d\by
		&\leq
		C\eps^2\sum_{k\neq j} \iint_{U\times \Omega} \frac{(\ln|\bx-\bx_k|-\ln|\by-\bx_k|)^2}{|\bx-\by|^2}d\bx d\by\\
		&=
		C\eps^2 \sum_{k\neq j} \iint_{(U-\bx_k)\times(\Omega-\bx_k)}\frac{1}{|\bx|^2}\frac{\ln^2\bigl|\frac{\by}{\bx}\bigr|}{\bigl(1-\bigl|\frac{\by}{\bx}\bigr|\bigr)^2}d\by d\bx\\
		&\leq
		C\eps^2\sum_{k\neq j} \int_{U-\bx_k}\frac{1}{|\bx|}\int_{0}^{\infty}\frac{\ln^2z}{(1-z)^2}dz d\bx.
	\end{align*}
	The inner integral is convergent and hence we are left with
	\begin{equation*}
		\eps^2\sum_{k\neq j} \int_{U-\bx_k}\frac{1}{|\bx|} d\bx
		\leq
		C\eps^2 \int_{e^{-2/\eps}}^{e^{-1/\eps}}z^{-1}dz= C\eps.
	\end{equation*}

	\medskip

	When $2|\alpha|=d-1$ and $d> 1$ the estimates for the difference quotient are a bit more technical. Similarly to above, Taylor's theorem combined with~\eqref{eq:phi_star_bound} yields
	\begin{align*}
		|D_{\bx_j}^\alpha \chi_\eps^*(\bx; \sx')-D^\alpha_{\bx_j} \chi_\eps^*(\by; \sx')|
		&= 
		\biggl|
			\sum_{|\beta|=1}\int_0^1 D^{\alpha+\beta}_{\bx_j}\chi_\eps^*(\bx+t(\by-\bx); \sx')(\by-\bx)^\beta dt
		\biggr|\\
		&\leq
		C\eps |\bx-\by| \sum_{k\neq j}\int_0^1 \frac{\1_{B^c_{e^{-2/\eps}}(\bx_k)}(\bx+t(\by-\bx))}{|\bx-t(\by-\bx)-\bx_k|^{|\alpha|+1}}dt.
	\end{align*}
	
	We estimate the integral
	\begin{equation*}
		\int_0^1 |\bx-t(\by-\bx)-\bx_k|^{-|\alpha|-1}dt.
	\end{equation*}

	Choosing coordinates in a plane containing $\bx_k, \bx$ and $\by$ such that $\bx_k=(0, 0)$, $\bx=(r_1, 0)$ and $\by=(r_2\cos(\theta), r_2 \sin(\theta))$ with $\theta \in [0, \pi)$ we can write this integral as
	\begin{align*}
		\int_0^1 &(((1-t)r_1-t r_2\cos\theta)^2+t^2r_2^2\sin^2\theta)^{-\frac{|\alpha|+1}{2}}dt\\
		&=
		\int_0^1 \Bigl(\Bigl(1-\frac{t r_2}{(1-t)r_1}\cos\theta\Bigr)^2+\frac{t^2r_2^2}{(1-t)^2r_1^2}\sin^2\theta\Bigr)^{-\frac{|\alpha|+1}{2}}dt\\
		&=
		\frac{1}{r_1 r_2^{|\alpha|}}\int_0^\infty \frac{(s+r_2/r_1)^{|\alpha|-1}}{((1+s \cos\theta)^2+s^2\sin^2\theta)^{\frac{|\alpha|+1}{2}}}ds\\
		&\leq
		\frac{1}{r_1 r_2^{|\alpha|}}\int_0^\infty \frac{(s+1)^{|\alpha|-1}}{((1-s)^2+2s(1+\cos\theta))^{\frac{|\alpha|+1}{2}}}ds\\
		&=:
		\frac{g(\theta)}{r_1 r_2^{|\alpha|}}
	\end{align*}
	The integral $g(\theta)$ tends to infinity in the limit $\theta\to \pi$. However, this corresponds to~$\bx$ and~$\by$ being far apart relative to their distance to the~$\bx_k$.

	When $\theta$ is far from $0$ we shall instead use the following bound which follows directly from the supremum bound in~\eqref{eq:phi_star_bound}
	\begin{equation}\label{eq:antipodal_bound1}
		|D_{\bx_j}^\alpha \chi_\eps^*(\bx; \sx')- D_{\bx_j}^\alpha \chi_\eps^*(\by; \sx')|
		\leq
		C \eps \sum_{k\neq j}\Biggl[\frac{\1_{B_{e^{-2/\eps}}^c(\bx_k)}(\bx)}{|\bx-\bx_k|^{|\alpha|}}+\frac{\1_{B_{e^{-2/\eps}}^c(\bx_k)}(\by)}{|\by-\bx_k|^{|\alpha|}}\Biggr]
	\end{equation}
	together with the fact that
	\begin{equation}\label{eq:antipodal_bound2}
		|\bx- \by| \geq \sin(\theta/2)\max\{|\bx-\bx_k|, |\by-\bx_k|\},
	\end{equation}
	where $\theta$ is the angle between the vectors $\by-\bx_k$ and $\bx-\bx_k.$ Note that the bound in~\eqref{eq:antipodal_bound1} does not capture the continuity of $D^\alpha\chi_\eps^*$ and hence cannot be sufficiently accurate for our purposes when $|\bx-\by|$ is small.

	\bigskip

	We are now ready to start estimating the $H^s$-seminorm of $\chi_\eps^*$. Using the same notation as in the $d=1$ case
	\begin{align*}
		\|\chi_\eps^*\|_{\dot H^s_{\bx_j}(\Omega)}^2
		&=
		\iint_{\Omega\times \Omega}\frac{|D_{\bx_j}^\alpha\chi_\eps^*(\bx; \sx')-D_{\bx_j}^\alpha\chi_\eps^*(\by; \sx')|^2}{|\bx-\by|^{d+1}}d\bx d\by\\
		&\leq
		2\iint_{U\times \Omega}\frac{|D_{\bx_j}^\alpha\chi_\eps^*(\bx; \sx')-D_{\bx_j}^\alpha\chi_\eps^*(\by; \sx')|^2}{|\bx-\by|^{d+1}}d\bx d\by,
	\end{align*}
	where we used that $|D_{\bx_j}^\alpha \chi_\eps^*(\bx; \sx')|=0$ for $\bx \in U_1 \cup U_2$, since $|\alpha|\geq 1$.

	To bound the integral we use the estimates derived earlier. Recalling that in the case under consideration $|\alpha|= \frac{d-1}{2}$ the derived bounds tells us that
	\begin{align*}
		&\iint_{U\times \Omega}\frac{|D_{\bx_j}^\alpha\chi_\eps^*(\bx; \sx')-D_{\bx_j}^\alpha\chi_\eps^*(\by; \sx')|^2}{|\bx-\by|^{d+1}}d\bx d\by\\
		&\leq
		 C \eps^2\sum_{k\neq j}
		 \iint_{U\times \Omega}
		 \frac{\min\bigl\{
		 \frac{g(\theta_k)^{2}}{|\bx-\bx_k|^{d-1}|\by-\bx_k|^{2}}, 
		 \frac{g(\theta_k)^{2}}{|\bx-\bx_k|^{2}|\by-\bx_k|^{d-1}}, 
		 \frac{|\bx-\by|^{-2}}{|\bx-\bx_k|^{d-1}}
		 +\frac{|\bx-\by|^{-2}}{|\by-\bx_k|^{d-1}}\bigr\}}{|\bx-\by|^{d-1}}d\bx d\by,
	\end{align*}
	here $\theta_k$ denotes the angle between the vectors $\bx-\bx_k$ and $\by-\bx_k$. For each fixed $\bx$ we rewrite the integral over $\Omega$ in spherical coordinates around $\bx_k$, oriented so that $\bx$ is located at the south pole. With $R=|\bx-\bx_k|, r=|\by-\bx_k|$ and $\theta_k$ as before, the integral becomes
	\begin{align*}
		&\iint_{U\times\Omega}\frac{\min\bigl\{
		\frac{g(\theta_k)^{2}}{R^{d-1}r^{2}}, 
		\frac{g(\theta_k)^{2}}{R^{2}r^{d-1}}, 
		\frac{|\bx-\by|^{-2}}{R^{d-1}}+
		\frac{|\bx-\by|^{-2}}{r^{d-1}}\bigr\}}{|\bx-\by|^{d-1}}d\bx d\by\\
		&\quad\leq 
			\int_U\int_0^{\infty}\int_0^\pi \int_{\S^{d-2}}
		\frac{
		\min\bigl\{ 
		\frac{g(\theta_k)^{2}}{R^{d-1}r^{2}}, 
		\frac{g(\theta_k)^{2}}{R^{2}r^{d-1}},
		\frac{|\bx-\by|^{-2}}{R^{d-1}}
		+\frac{|\bx-\by|^{-2}}{r^{d-1}}
		\bigr\}}
		{((R-r \cos\theta_k)^2+r^2\sin^2 \theta_k |\hat\theta|^2)^{(d-1)/2}}
		r^{d-1}\sin^{d-2}\theta_k dr d\theta_k d S(\hat\theta) d\bx\\
		&\quad=
			C\int_U\int_0^{\infty}\int_0^\pi
		\frac{
		\min\bigl\{
		\frac{g(\theta_k)^{2}}{R^{d-1}r^{2}}, 
		\frac{g(\theta_k)^{2}}{R^{2}r^{d-1}},
		\frac{|\bx-\by|^{-2}}{R^{d-1}}+
		\frac{|\bx-\by|^{-2}}{r^{d-1}}
		\bigr\}}
		{((R-r \cos\theta_k)^2+r^2\sin^2 \theta_k)^{(d-1)/2}}
		r^{d-1}\sin^{d-2}\theta_k dr d\theta_k d\bx.
	\end{align*}

	For $\theta \in [\pi/2, \pi]$ we use the bounds in~\eqref{eq:antipodal_bound1},~\eqref{eq:antipodal_bound2}:
	\begin{align*}
		\int_U\int_0^{\infty}\int_{\pi/2}^\pi&
		\frac{
		R^{-d+1}+r^{-d+1}}
		{\sin^{d+1}(\theta_k/2) \max\{R, r\}^{d+1}}
		r^{d-1}\sin^{d-2}\theta_k dr d\theta_k d\bx\\
		&=
		\int_U\int_R^{\infty}\int_{\pi/2}^\pi
		\frac{
		R^{-d+1}+r^{-d+1}}
		{\sin^{d+1}(\theta_k/2)  r^{2}}\sin^{d-2}(\theta_k)
		dr d\theta_k d\bx\\
		&\quad+
		\int_U\int_0^{R}\int_{\pi/2}^\pi
		\frac{
		R^{-d+1}+r^{-d+1}}
		{\sin^{d+1}(\theta_k/2) R^{d+1}}
		r^{d-1}\sin^{d-2}(\theta_k)dr d\theta_k d\bx\\
		&\leq
		C\int_U\int_R^{\infty}
		\frac{
		R^{-d+1}+r^{-d+1}}
		{r^{2}}
		dr d\bx
		+
		C\int_U\int_0^{R}
		\frac{
		R^{-d+1}+r^{-d+1}}
		{R^{d+1}}
		r^{d-1}dr d\bx\\
		&=
		C \int_U R^{-d}d\bx\\
		&\leq
		C \int_{e^{-2/\eps}}^{e^{-1/\eps}}R^{-1}dR
		 = C\eps^{-1}.
	\end{align*}
	Thus this part of the integral is $O(\eps^{-1})$.

	\medskip

	What remains is to bound the integral when $r\geq 0$ and $\theta_k\in [0, \pi/2).$ To accomplish this we shall use the bound for the difference of the derivatives derived earlier. Note that since $\theta_k<\pi/2$ we can replace the factor $g(\theta_k)$ by a constant without any loss. Using that $|\bx-\by|^2=R^2+r^2-2r R\cos\theta\geq \max\{(R-r)^2, 2r R(1-\cos\theta_k)\}$ we for any fixed $\mu\in (0, 1)$ find
	\begin{align*}
		\int_U\int_{0}^{\infty}&\int_0^{\pi/2}
		\frac{
		\min\bigl\{ 
		\frac{1}{R^{d-1}r^{2}}, 
		\frac{1}{R^{2}r^{d-1}}
		\bigr\}}
		{((R-r \cos\theta_k)^2+r^2\sin^2 \theta_k)^{(d-1)/2}}
		r^{d-1}\sin^{d-2}\theta_k dr d\theta_k d\bx\\
		&=
		\int_U\int_{0}^{R}\int_0^{\pi/2}
		((R-r \cos\theta_k)^2+r^2\sin^2 \theta_k)^{-(d-1)/2}
		R^{-d+1}r^{d-3}\sin^{d-2}\theta_k dr d\theta_k d\bx\\
		&\quad+
		\int_U\int_{R}^{\infty}\int_0^{\pi/2}
		((R-r \cos\theta_k)^2+r^2\sin^2 \theta_k)^{-(d-1)/2}R^{-2}\sin^{d-2}\theta_k dr d\theta_k d\bx\\
		&\leq
		\int_U\int_{0}^{R}
		(R-r)^{-\mu}r^{(d-5+\mu)/2} R^{-(3d-3-\mu)/2}  dr  d\bx 
		\int_0^{\pi/2}\frac{\sin^{d-2}\theta_k}{(1-\cos\theta_k)^{(d-1-\mu)/2}}d\theta_k\\
		&\quad+
		\int_U\int_{R}^{\infty}
		(r-R)^{-\mu}r^{-(d-1-\mu)/2} R^{-(d-3-\mu)/2} dr d\bx
		\int_0^{\pi/2}\frac{\sin^{d-2}\theta_k}{(1-\cos\theta_k)^{(d-1-\mu)/2}}d\theta_k\\
		&\leq
		C\int_U (R^{-d}+R^{-d+3}) d\bx\\
		&\leq
		C\int_{e^{-2/\eps}}^{e^{-1/\eps}} (R^{-1}+R^2) dR
		=
		C \eps^{-1}.
	\end{align*}
	Consequently, also this part of the integral is $O(\eps^{-1})$ which completes the proof.
\end{proof}

\begin{lemma}\label{lem:inclusion_of_spaces_W}
	For all $s > 0$ it holds that $\cH^{s,N}_W(\R^d) \subseteq \cH^{s,N}_0(\R^d)$.
\end{lemma}
\begin{proof}
	Take $\Psi \in H^s(\R^{dN})$ s.t. $\int W_s|\Psi|^2 < \infty$ and
	let $\Psi_\eps := \chi_\eps\Psi$. Since $\Psi_\eps$ is supported away from
	$\bDelta_\eps := \bDelta + B_\eps(0)$ and thus may be approximated in 
	$C_c^\infty(\R^{dN} \setminus \bDelta)$, it is sufficient to prove that
	$\|\Psi-\Psi_\eps\|_{H^{s}(\R^{dN})} \to 0$ to conclude the lemma.
	We have by dominated convergence
	$$
		\|\Psi-\Psi_\eps\|_{L^2(\R^{dN})}^2 \lesssim \int_{\scriptbDelta_\eps \cap \,\R^{dN}} |1-\chi_\eps|^2 |\Psi|^2 \to 0,
	$$
	while for $\alpha \neq 0$
	$$
		D^\alpha_{\bx_j}((1-\chi_\eps)\Psi) = \sum_{0 \le \beta \le \alpha} D^\beta_{\bx_j}(1-\chi_\eps) D^{\alpha-\beta}_{\bx_j}\Psi,
	$$
	so for $s=m+\sigma$, $|\alpha| = m$, $0 \le \sigma < 1$ (for $\sigma=0$ we replace by $L^2$)
	\begin{align*}
		\|\Psi-\Psi_\eps\|_{\dot{H}^{s,N}(\R^d)} 
		&\lesssim \sum_{j,\alpha} \|(1-\chi_\eps)D^\alpha_{\bx_j}\Psi\|_{\dot{H}^{\sigma,N}(\R^d)} 
		+ \sum_{j,\alpha} \sum_{0<\beta<\alpha} \|(D^\beta_{\bx_j} \chi_\eps) (D^{\alpha-\beta}_{\bx_j}\Psi)\|_{\dot{H}^{\sigma,N}(\R^d)} \\
		&+ \sum_{j,\alpha} \|(D^\alpha_{\bx_j} \chi_\eps) \Psi\|_{\dot{H}^{\sigma,N}(\R^d)}.
	\end{align*}
	We may estimate as in the proof of Lemma~\ref{lem:norm_of_cutoff},
	$$
		\|(1-\chi_\eps)D^\alpha\Psi\|_{\dot{H}^{\sigma,N}(\R^d)}^2
		\lesssim \sum_j\sum_{k\neq j} \int_{\R^{d(N-1)}} \|\Psi\|_{H^{s}_{\bx_j}(B_\eps(\bx_k))}^2
		\to 0,
	$$
	\begin{align*}
		\|\Psi D^\alpha \chi_\eps\|_{\dot{H}^\sigma_{\bx_j}(\R^d)}^2 
		&= \iint_{\R^d \times \R^d} \dfrac{|\Psi D^\alpha\chi_\eps(\bx;\sx') - \Psi D^\alpha\chi_\eps(\by;\sx')|^2}{|\bx-\by|^{d+2\sigma}} \,d\bx d\by \\
		&\lesssim I_\alpha + \eps^{-2|\alpha|-2} \sum_{k,l \neq j} \iint_{B_{2\eps}(\bx_k) \times B_{2\eps}(\bx_l)} 
			|\Psi(\bx;\sx')|^2|\bx - \by|^{-d-2\sigma+2} d\bx d\by \\
		&\lesssim I_\alpha + \eps^{-2|\alpha|-2} \sum_{k \neq j} \int_{B_{2\eps}(\bx_k)}
			|\Psi(\bx;\sx')|^2 \int_{B_{4\eps}(\bx)} |\bx - \by|^{-d-2\sigma+2} d\by \,d\bx \\
		&\lesssim I_\alpha + \eps^{-2|\alpha| - 2\sigma} \sum_{k \neq j} 
			\int_{B_{2\eps}(\bx_k)} |\Psi(\bx;\sx')|^2 \,d\bx,
	\end{align*}
	where
	\begin{equation*}
		I_\alpha
		= \sum_{k \neq j} \int_{B_{2\eps}(\bx_k)}
			|D^\alpha\chi_\eps(\bx;\sx')|^2 \int_{B_{4\eps}(\bx_k)} \dfrac{|\Psi(\bx;\sx') - \Psi(\by;\sx')|^2}{|\bx-\by|^{d+2\sigma}} d\by \,d\bx.
	\end{equation*}

	For the highest-order derivatives $2|\alpha|=2s-2\sigma$:
	$$
		\|\Psi D^\alpha \chi_\eps\|_{\dot{H}^{\sigma,N}(\Omega)}^2 
		\lesssim \int_{\Omega^{N-1}} I_\alpha + \eps^{-2s} 
			\int_{\scriptbDelta_{2\eps}} |\Psi(\sx)|^2 \,d\sx
		\lesssim \int_{\Omega^{N-1}} I_\alpha +  
			\int_{\scriptbDelta_{2\eps}} W_s(\sx) |\Psi(\sx)|^2 \,d\sx,
	$$
	where the last term tends to zero as $\eps \to 0$ by dominated convergence.
	
	For $I_\alpha$ we have that
	\begin{align*}
		I_\alpha
		&= \sum_{k \neq j} \int_{B_{2\eps}(\bx_k)}
			|D^\alpha\chi_\eps(\bx;\sx')|^2 \int_{B_{4\eps}(\bx_k)} \dfrac{|\Psi(\bx;\sx') - \Psi(\by;\sx')|^2}{|\bx-\by|^{d+2\sigma}} d\by \,d\bx \\
		&\lesssim \eps^{-2|\alpha|} \sum_{k \neq j} \iint_{B_{4\eps}(\bx_k) \times B_{4\eps}(\bx_k)}
			\dfrac{|\Psi(\bx;\sx') - \Psi(\by;\sx')|^2}{|\bx-\by|^{d+2\sigma}} d\by \,d\bx \\
		&\lesssim \eps^{-2|\alpha|} \sum_{k \neq j} \|\Psi\|_{\dot{H}^\sigma_{\bx_j}(B_{4\eps}(\bx_k))}^2.
	\end{align*}
	By interpolation of Sobolev spaces and scaling we have
	for $C = C(d,\sigma,m)>0$
	\begin{equation*}
		\|\Psi\|_{\dot{H}^\sigma_{\bx_j}(B_{4\eps}(\bx_k))}^2 
		\le C\eps^{2m} \bigl(\|\Psi\|_{\dot{H}^{s}_{\bx_j}(B_{4\eps}(\bx_k))}^2+ \|\Psi \sqrt{W_s}\|^2_{L^2_{\bx_j}(B_{4\eps}(\bx_k))}\bigr),
	\end{equation*}
	and thus by dominated convergence 
	\begin{equation*}
	\int_{\R^{d(N-1)}} I_\alpha \lesssim \sum_{k\neq j}\int_{\R^{d(N-1)}} \bigl(\|\Psi\|_{\dot{H}^{s}_{\bx_j}(B_{4\eps}(\bx_k))}^2+\|\Psi\sqrt{W_s}\|^2_{L^2_{\bx_j}(B_{4\eps}(\bx_k))}\bigr) \to 0,
	\end{equation*}
	for $|\alpha| = s-\sigma$.
	Similarly, for the lower-order mixed terms
	\begin{align*}
		\|(D^{\alpha-\beta}_{\bx_j}\Psi)(D^\beta_{\bx_j} \chi_\eps)\|_{\dot{H}^{\sigma}_{\bx_j}(\R^d)}
		&\lesssim \eps^{-2|\beta|} \sum_{k \neq j} \|D^{\alpha-\beta}_{\bx_j}\Psi\|_{\dot{H}^\sigma_{\bx_j}(B_{4\eps}(\bx_k))}^2\\
		&\quad 
		+ \eps^{-2|\beta|-2\sigma} \sum_{k \neq j} \|D^{\alpha-\beta}_{\bx_j}\Psi\|_{L^2_{\bx_j}(B_{4\eps}(\bx_k))}^2 \\
		&\lesssim 
		\sum_{k \neq j} \bigl(\|\Psi\|_{\dot{H}^{s}_{\bx_j}(B_{4\eps}(\bx_k))}^2
		+
		\|\Psi\sqrt{W_s}\|_{L^2_{\bx_j}(B_{4\eps}(\bx_k))}^2\bigr),
	\end{align*}
	which implies that also
	$\|(D^{\alpha-\beta} \Psi)(D^\beta \chi_\eps)\|_{\dot{H}^{\sigma,N}(\R^d)} 
		\to 0$.
\end{proof}

\begin{lemma}\label{lem:inclusion_of_spaces_0}
	For all $0 < 2s \le d$ it holds that 
	$\cH^{s,N}_0(\R^d) = \cH^{s,N}(\R^d) = H^s(\R^{dN})$.
\end{lemma}
\begin{proof}
	As mentioned above, combining Lemma~\ref{lem:inclusion_of_spaces_W} with the Hardy--Rellich inequality implies the claim when $0<2s<d$. For $2s=d$ we argue as follows.

	It suffices to prove that $C_c^\infty(\R^{dN} \setminus \bDelta)$
	is dense in~$H^s$, and moreover, 
	using that $C_c^\infty(\R^{dN})$ is dense in $H^s$,
	it suffices to prove that if $\Psi \in C_c^\infty(\R^{dN})$
	then $\Psi_\eps := \chi^*_\eps\Psi \to \Psi$ in~$H^s$ as $\eps \to 0$.
	Clearly
	$$
		\|\Psi-\Psi_\eps\|_{L^2(\R^{dN})}^2 
		\lesssim \int_{\scriptbDelta_\eps \cap\, \mathrm{supp} \Psi} |1-\chi^*_\eps|^2 \to 0.
	$$
	Moreover, by Lemma~\ref{lem:norm_of_cutoff} and arguing as in the proof of Lemma~\ref{lem:inclusion_of_spaces_W}
	\begin{equation*}
		\|\Psi-\Psi_\eps\|_{\dot{H}^{s,N}(\R^{d})}^2 
		\lesssim \eps \to 0.\qedhere
	\end{equation*}
	
\end{proof}

The above generalizes the case $d=2$ and $s=1$ where it is well known that
hard-core bosons have non-extensive energy in the dilute limit \cite{LieYng-01}
and thus that a Lieb--Thirring inequality of the type \eqref{eq:LT-diagonal} cannot hold.
See also \cite{Svendsen-81} for generalizations with integer $s$.


\end{document}